\documentclass[a4paper,reqno,10pt]{amsart} 

\date{\today}

\usepackage{amsmath}
\usepackage{amsfonts}
\usepackage{amssymb}
\usepackage{amsthm}
\usepackage{amstext}
\usepackage{enumerate}
\usepackage{url}

\newcommand{\bx}{\mathbf{x}} 
\newcommand{\by}{\mathbf{y}}
\newcommand{\bz}{\mathbf{z}}
\newcommand{\bp}{\mathbf{p}}

\newcommand{\nz}{\mathbb{N}}
\newcommand{\cz}{\mathbb{C}}
\newcommand{\rz}{\mathbb{R}}
\newcommand{\N}{\mathbb{N}}

\newcommand{\R}{\mathbb{R}}

\newcommand{\Tr}{\operatorname{Tr}}

\newcommand{\supp}{{\operatorname{supp}}}
\newcommand{\Span}{{\operatorname{span}}}

\newtheorem{lemma}{Lemma}
\newtheorem{theorem}{Theorem}
\newtheorem{proposition}{Proposition}
\newtheorem{remark}{Remark}

\newtheorem*{acknowledgement}{Acknowledgement}

\newenvironment{pf*}[1]{\par\medskip\noindent\textit{#1}\,:}{\hspace*{\fill}\qed\medskip\par\noindent}   

\title{Hartree-Fock theory for pseudorelativistic atoms}

\author[A. Dall'Acqua, T. \O stergaard S\o rensen, and E. Stockmeyer]{Anna
  Dall'Acqua, Thomas \O stergaard S\o rensen, and Edgardo Stockmeyer}

\thanks{\copyright\ 2007 by the
       authors. This article may be reproduced in its entirety for
       non-commercial purposes.}

\address[Anna Dall'Acqua]
{Zentrum Mathematik der
Technischen Universit\"at M\"unchen,
Boltzmannstrasse 3
D-85748 Garching, Germany.}
\email{dallacqu@ma.tum.de}

\address[Thomas \O stergaard S\o rensen]
{Department of Mathematical Sciences,
           Aalborg University,
           Fredrik Bajers Vej 7G,
           DK-9220 Aalborg East, Denmark.}
\email{sorensen@math.aau.dk}

\address[Edgardo Stockmeyer]
{Mathematisches Institut,
Universit\"at M\"unchen,
Theresien\-stra\ss e 39,
D-80333 Munich, Germany.}
\email{stock@math.lmu.de}

\begin{document}

 \begin{abstract}  
  We study the Hartree-Fock model for pseudorelativistic atoms, that
  is, atoms where the kinetic energy of the electrons is given by the 
   pseudo\-relativistic operator $\sqrt{(|\bp|c)^2+(mc^2)^2}-mc^2$. We prove the
   existence of a Hartree-Fock minimizer, and prove regularity away
   from the nucleus and pointwise exponential decay of the
   corresponding orbitals.
 \end{abstract}

\maketitle

\section{Introduction and results}

We consider a model for an atom with $N$ electrons and nuclear charge
$Z$, where the kinetic energy of the electrons is described by
the expression $\sqrt{(|\bp|c)^2+(mc^2)^2}-mc^2$. This model takes into
account some (kinematic) relativistic effects; in units where
\(\hbar=e=m=1\), the Hamiltonian becomes
\begin{align}\label{Hamiltonian}
  H=H_{\rm rel}(N,Z,\alpha)&=\sum_{j=1}^{N}\Big\{ \sqrt{-\alpha ^{-2}\Delta _{j}+\alpha ^{-4}}
  -\alpha ^{-2}-\frac{Z}{|\bx_{j}|}\Big\}
  +\sum_{1\leq i<j\leq N}\frac{1}{|\bx_{i}-\bx_{j}|}  \nonumber \\
  &=\sum_{j=1}^{N}\alpha ^{-1}\Big\{T(-{\rm i}\nabla_{j})-V({\bx}_{j})
  \Big\}
  +\sum_{1\leq i<j\leq N}\frac{1}{|\bx_{i}-\bx_{j}|}\,,  
\end{align}
with
\(T({\bp})=E({\bp})-\alpha^{-1}=\sqrt{|{\bp}|^2+\alpha^{-2}}-\alpha^{-1}\)  
and \(V({\bx})=Z\alpha/|{\bx}|\). Here, $\alpha $ is Sommerfeld's fine
structure constant; physically, \(\alpha\simeq1/137.036\).

The operator $H$ acts on a dense subspace of the
$N$-particle Hilbert space $\mathcal{H}_{F}=\wedge 
_{i=1}^{N}L^{2}(\mathbb{R}^{3};\mathbb{C}^{q})$ of antisymmetric
functions, where $q$ is the number of 
spin states. It is bounded from below on this subspace (more details below).

The \textit{(quantum) ground state energy} is the infimum of the
spectrum of $H$ considered as an operator acting on
$\mathcal{H}_{F}$: 
\begin{equation*}
  E^{{\rm QM}}(N,Z,\alpha ):=\inf \sigma _{\mathcal{H}_{F}}(H)
  =\inf\{\,\mathfrak{q}(\Psi,\Psi)\,|\, \Psi\in \mathcal{Q}(H),
  \langle\Psi,\Psi\rangle=1\} \,, 
\end{equation*}
where \(\mathfrak{q}\) is the quadratic form defined by \(H\), and
\(\mathcal{Q}\) the corresponding form domain (see below); \(\langle\
,\ \rangle\) is the scalar product in \(\mathcal{H}_{F}\subset
L^2(\rz^{3N};\cz^{q^N})\).  

In the Hartree-Fock approximation, instead of minimizing the
functional \(\mathfrak{q}\) in the entire $N$-particle space
\(\mathcal{H}_{F}\), one 
restricts to wavefunctions \(\Psi\) which are pure wedge products,
also called Slater determinants: 
\begin{align}\label{slater}
  \Psi (\bx_{1},\sigma_{1};\bx_{2},\sigma _{2};
  \dots;\bx_{N},\sigma_{N})
  =\frac{1}{\sqrt{N!}}\,\det(u_{i}(\bx_{j},\sigma_{j}))_{i,j=1}^{N}\,, 
\end{align}
with $\{u_{i}\}_{i=1}^N$ orthonormal in
$L^{2}(\mathbb{R}^{3};\mathbb{C}^{q})$ (called {\it orbitals}). Notice 
that this way, \(\Psi\in\mathcal{H}_{F}\) and
$\|\Psi\|_{L^{2}(\mathbb{R}^{3N};\mathbb{C}^{q^N})}=1$.   

The {\it Hartree-Fock ground state energy} is the infimum of the
quadratic form \(\mathfrak{q}\) defined by \(H\) over such Slater
determinants:  
\begin{align}\label{eq:HF-energy}
  E^{{\rm HF}}(N,Z,\alpha):=
  \inf \{\,\mathfrak{q}(\Psi,\Psi) \,|\, \Psi 
  \text{ Slater determinant} \}\,.
\end{align}

For the non-relativistic Hamiltonian, 
\begin{align}\label{eq:non-rel}
  H_{\rm cl}(N,Z)=\sum_{j=1}^{N}\Big\{-\frac{1}{2}\Delta_j
  -\frac{Z}{|\bx_{j}|}\Big\} 
  +\sum_{1\leq i<j\leq N}\frac{1}{|\bx_{i}-\bx_{j}|}\,,
\end{align}
the mathematical theory of this approximation has been much studied,
the groundbreaking work being that of Lieb and
Simon~\cite{LiebSimonHF}; see also \cite{Lions} for work on excited
states. For a comprehensive 
discussion of Hartree-Fock (and other) approximations in quantum
chemistry, and an extensive literature list, we refer to \cite{LeBrisLions}. 

The aim of the present paper is to study the  Hartree-Fock
approximation for the pseudorelativistic operator \(H\) in
\eqref{Hamiltonian}. 

We turn to the precise description of the problem. 
The one-particle operator \(h_0=T(-{\rm
  i}\nabla)-V({\bf x})\) is bounded from below (by
\(\alpha^{-1}[(1-(\pi Z\alpha/2)^2)^{1/2}-1]\))
if and only if 
$Z\alpha \leq 2/\pi $ (see \cite{Herbst}, \cite[5.33 p.\ 307]{Kato}, and
\cite{Weder}; we shall have nothing further to say on the critical case
\(Z\alpha=2/\pi\)). More precisely, if \(Z\alpha<1/2\), then \(V\) is
a small {\it operator} pertubation of \(T\). In fact \cite[Theorem 2.1
c)]{Herbst}, \(\big\||\bx|^{-1}(T(-{\rm 
  i}\nabla)+1)^{-1}\big\|_{\mathcal{B}(L^2(\rz^3))}=2\). As a
consequence, \(h_0\) is selfadjoint with
\(\mathcal{D}(h_0)=H^1(\rz^3;\cz^q)\) when \(Z\alpha<1/2\).
It is essentially
selfadjoint on \(C_0^{\infty}(\rz^3;\cz^q)\) when \(Z\alpha\leq 1/2\).

If, on the other hand, \(1/2\le
Z\alpha<2/\pi\), then \(V\) is only a
small {\it form} pertubation of \(T\): Indeed \cite[5.33
p.\ 307]{Kato}, 
\begin{align}
  \label{eq:Kato}
  \int_{\rz^3}\frac{|f(\bx)|^2}{|\bx|}\,d\bx
  \le \frac{\pi}{2}\int_{\rz^3}|\bp||\hat{f}(\bp)|^2\,d\bp\  \text{
    for }\ 
   f\in H^{1/2}(\rz^3)\,,
\end{align}
where \(\hat{f}\) denotes the Fourier transform of \(f\). Hence, the
quadratic form $\mathfrak{v}$ given by
\begin{equation}\label{b}
  \mathfrak{v}[u,v]:=(V^{1/2}u, V^{1/2} v)\ \text{ for }\ u,v \in
  H^{1/2}(\mathbb{R}^3;\mathbb{C}^q)\,
\end{equation}
(multiplication by \(V^{1/2}\) in each component)
is well defined (for all values of \(Z\alpha\)). Here,  $(\ ,\ )$
denotes the scalar product in 
$L^2(\mathbb{R}^3;\mathbb{C}^q)$. 
Let $\mathfrak{e}$ be the
quadratic form with domain 
$H^{1/2}(\mathbb{R}^3;\mathbb{C}^q)$ given by  
\begin{equation}\label{a}
  \mathfrak{e}[u,v]:=(E(\bp)^{1/2}u,E(\bp)^{1/2}v)\  \text{
    for }\ u,v \in H^{1/2}(\mathbb{R}^3;\mathbb{C}^q)\,. 
\end{equation}
By abuse of notation, we write \(E(\bp)\) for the (strictly positive)
operator \(E(-{\rm i}\nabla)=\sqrt{-\Delta+\alpha^{-2}}\).
Then, using \eqref{eq:Kato} and that \(|\bp|\le E(\bp)\), 
\begin{align}\label{Kato's ineq}
  \mathfrak{v}[u,u] <
  \mathfrak{e}[u,u] \ \text{ for }\ u \in
  H^{1/2}(\mathbb{R}^3;\mathbb{C}^q) 
  \ \text{ if }\ Z\alpha<2/\pi\,.
\end{align}
Hence, by the KLMN 
theorem \cite[Theorem X.17]{RS2}, there exists a unique self-adjoint 
operator $h_0$ whose quadratic form domain is $H^{1/2}(\mathbb{R}^3;
\mathbb{C}^q)$ such that (with
\(\mathfrak{t}=\mathfrak{e}-\alpha^{-1}\)) 
\begin{align}\label{1p}
  (u,h_0v)=\mathfrak{t}[u,v]-\mathfrak{v}[u,v] 
  \ \text{ for }\  u,v \in
  H^{1/2}(\mathbb{R}^3; \mathbb{C}^q)\,, 
\end{align}
and $h_0$ is bounded below by \(-\,\alpha^{-1}\). Moreover, if
$Z\alpha<2/\pi $ then the spectrum of \(h_0\) is discrete in $[-\alpha 
^{-1},0)$ and absolutely continuous in $[0,\infty )$ \cite[Theorems
2.2 and 2.3]{Herbst}.

As for the \(N\)-particle operator in \eqref{Hamiltonian}, when
\(Z\alpha<2/\pi\), \eqref{eq:Kato} implies that the quadratic form
\begin{align*}
  \mathfrak{q}&(\Psi,\Phi)
  =\sum_{j=1}^{N}
  \Big\{\langle\,E(\bp_j)^{1/2}\Psi,E(\bp_j)^{1/2}\Phi\,\rangle
  -\alpha^{-1}\langle\Psi,\Phi\rangle
  -\langle \,V(\bx_j)^{1/2}\Psi,V(\bx_j)^{1/2}\Phi\,\rangle\Big\}
  \nonumber
  \\&+\sum_{1\le i<j\le N}
  \langle\,|\bx_i-\bx_j|^{-1/2}\Psi,
  |\bx_i-\bx_j|^{-1/2}\Phi\rangle\ ,\quad
  \Psi,\Phi\in
  \bigwedge_{i=1}^{N}H^{1/2}(\rz^3;\cz^q)\,,
\end{align*}
is well-defined, closed, and bounded from below. The operator \(H\)
can then be defined as the corresponding (unique) self-adjoint
operator. It satisfies 
\begin{align*}
  \bigwedge_{i=1}^{N}H^{1}(\rz^3;\cz^q)
  \subset
  \mathcal{D}(H)\subset\mathcal{Q}(H)=\bigwedge_{i=1}^{N}H^{1/2}(\rz^3;\cz^q)\,, 
  \\
   \mathfrak{q}(\Psi,\Phi)=\langle \Psi,H\Phi\rangle\,, \quad \Phi\in
  \mathcal{D}(H)\,,\quad \Psi\in \mathcal{Q}(H)\,.
\end{align*}
For
\(Z\alpha<1/2\),
\(\mathcal{D}(H)=\wedge_{i=1}^{N}H^{1}(\rz^3;\cz^q)\). 
All this follows from (the statements and proofs of)
\cite[Theorem~X.17]{RS2} and \cite[Theorem~VIII.15]{RS1}.
See \cite{LiYau88} for further references on \(H\).
We shall not have anything further to say on \(H\) in this paper,
however, but will only study the Hartree-Fock problem mentioned
above. We now discuss this in more detail.

It is convenient to use
the one-to-one correspondence between Slater determinants and
projections onto finite dimensional subspaces of $L^2(\mathbb{R}^3;
\mathbb{C}^q)$. Indeed, if $\Psi$ is given by \eqref{slater} with
\(\{u_i\}_{i=1}^{N}\subset H^{1/2}(\rz^3;\cz^q)\), orthonormal in
\(L^{2}(\mathbb{R}^{3};\mathbb{C}^{q})\), and
$\gamma $ is the projection onto the subspace spanned by
$u_{1},\ldots,u_{N}$, then the kernel of \(\gamma\) is given by
\begin{equation}\label{kerFirst}
  \gamma(\bx,\sigma;\by,\tau)
  =\sum_{j=1}^{N}u_{j}(\bx,\sigma)\overline{u_{j}(\by,\tau)}\,.
\end{equation}
Let $\rho_{\gamma}\in L^1(\rz^3)$ denote the
$1$-particle density associated to $\gamma$ given by  
\begin{align*}
  \rho_{\gamma}(\bx)=\sum_{\sigma=1}^{q}\gamma(\bx,\sigma;\bx,\sigma)
  =\sum_{\sigma=1}^{q}\sum_{j=1}^{N}|u_{j}(\bx,\sigma)|^{2}\,.  
\end{align*}
Then the energy expectation of \(\Psi\) depends only on
$\gamma$, more 
precisely, 
\begin{equation*}
  \mathfrak{q}(\Psi,\Psi)=
  \langle \Psi ,H\Psi \rangle =\mathcal{E}^{\rm HF}(\gamma )\,,
\end{equation*}
where $\mathcal{E}^{\rm HF}$ is the Hartree-Fock energy functional
defined by 
\begin{align}\label{eq:HF-functional}
  \mathcal{E}^{\rm HF}&(\gamma)=\alpha^{-1}
  \big\{\Tr[
  E(\bp)\gamma]-\alpha^{-1}\Tr[\gamma]-\Tr[ 
  V\gamma]\big\} + \mathcal{D} 
  (\gamma) - \mathcal{E}x(\gamma)\,.
\end{align}
Here, 
 \begin{equation*}\label{kintrA}
   \Tr[E(\bp)\gamma]:=\sum_{j=1}^N\mathfrak{e}[u_j,u_j]\ , \quad 
  \Tr[V\gamma]:=\sum_{j=1}^N\mathfrak{v}[u_j,u_j]
  =Z\alpha\int_{\rz^3}\frac{\rho_\gamma(\bx)}{|\bx|}\,d\bx\,,
\end{equation*}
$\mathcal{D}(\gamma)$ is the {\it direct} Coulomb energy,
\begin{equation}\label{def:DirCoul}
  \mathcal{D}(\gamma )=\frac{1}{2}\int_{\mathbb{R}^{3}}\int_{\mathbb{R}^{3}}
  \frac{\rho_{\gamma}(\mathbf{x})\rho_{\gamma}(\by)}{\vert 
  \mathbf{x}-\by\vert}\,d\mathbf{x}\,d\by\,,
\end{equation}
and $\mathcal{E}x(\gamma) $ is the {\it exchange} Coulomb energy,
\begin{equation*}
  \mathcal{E}x(\gamma) =\frac{1}{2}\sum_{\sigma,\tau=1}^{q}\int_{\mathbb{R}^{3}}\int_{
   \mathbb{R}^{3}}\frac{\vert
  \gamma(\bx,\sigma;\by,\tau)\vert ^{2}}{\vert\mathbf{x}
  -\by\vert}\,d\mathbf{x}\,d\by\,.
\end{equation*}
This way,
\begin{align}\label{eq:HF-energy2}
  E^{{\rm HF}}&(N,Z,\alpha)=
  \inf \{\,\mathcal{E}^{\rm HF}(\gamma) \,|\,
  \gamma\in\mathcal{P}\,\}\,,
  \\
  \nonumber
  \mathcal{P}&=\{ \gamma: L^{2}(\mathbb{R}^{3};\mathbb{C}^{q})\rightarrow
  L^{2}(\mathbb{R}^{3};\mathbb{C}^{q})\,|\, \gamma \
  \, \text{projection onto } \Span\{u_1,\ldots,u_N\}, 
  \\ 
  &\qquad\qquad\qquad\qquad\qquad\qquad\qquad\qquad\qquad u_i\in
  H^{1/2}(\rz^3;\cz^q), (u_i,u_j)=\delta_{i,j}\}\,.
  \nonumber
\end{align}
(Notice that if one of the orbitals \(u_i\) of \(\gamma\) is not in
\(H^{1/2}(\rz^3;\cz^q)\), then \(\mathcal{E}^{\rm
  HF}(\gamma)=+\infty\) (since \(Z\alpha<2/\pi\)).)

We now extend the definition of the Hartree-Fock energy
functional \(\mathcal{E}^{\rm HF}\), in order to turn the minimization
problem \eqref{eq:HF-energy2} (that is, \eqref{eq:HF-energy}) into a
convex problem. 

A {\it density matrix} $\gamma: L^{2}(\mathbb{R}^{3};\mathbb{C}^{q})\rightarrow
L^{2}(\mathbb{R}^{3};\mathbb{C}^{q})$ is a self-adjoint trace class operator that
satisfies the operator inequality $0\leq \gamma \leq \mbox{\textrm{Id}}$.
A density matrix $\gamma$ has the integral kernel 
\begin{equation}\label{ker}
  \gamma(\bx,\sigma;\by,\tau)
  =\sum_{j}\lambda 
  _{j}u_{j}(\bx,\sigma)\overline{u_{j}(\by,\tau)}\,,
\end{equation} 
where $\lambda _{j},u_{j}$ are the eigenvalues and corresponding
eigenfunctions of $\gamma $. We choose the $u_j$'s to be orthonormal
in $L^2(\mathbb{R}^3;\mathbb{C}^q)$. As before, let $\rho_{\gamma}\in
L^1(\rz^3)$ denote the 
$1$-particle density associated to $\gamma$ given by  
\begin{equation}\label{def:rho}
  \rho _{\gamma }(\mathbf{x}) =\sum_{\sigma =1}^{q}\sum_{j}\lambda
  _{j}\left\vert u_{j}(\mathbf{x},\sigma )\right\vert ^{2}\,.
\end{equation}
Define
\begin{align}\label{A} 
  \mathcal{A}:=\big\{ \gamma \text{ density matrix }\big|\,\Tr\big[
  E(\textbf{p}) \gamma\big]<+\infty \,\big\} \,,  
\end{align} 
where, by definition, for \(\gamma\) written as in \eqref{ker}, 
\begin{equation}\label{kintr}
  \Tr[E(\bp)\gamma]:=\sum_{j}\lambda _{j} \mathfrak{e}[u_j,u_j]\,.
\end{equation}
Notice that if $\gamma \in \mathcal{A}$ then
all the terms in $\mathcal{E}^{\rm HF}(\gamma )$ (see
\eqref{eq:HF-functional}) are finite. Indeed, for $\gamma \in
\mathcal{A}$ and written as in \eqref{ker}, 
\begin{equation}\label{eq:potntialQuad}
  \Tr[V\gamma]:
  =\sum_{j}\lambda _{j} \mathfrak{v}[u_j,u_j]
  =Z\alpha\int_{\rz^3}\frac{\rho_\gamma(\bx)}{|\bx|}\,d\bx\,
\end{equation}
is finite, due to \eqref{Kato's ineq}. In particular, 
\begin{align}\label{eq:prop_u'i}
  u_j\in
  H^{1/2}(\rz^3;\cz^q)\subset L^3(\rz^3;\cz^q)\,,
\end{align}
the last inclusion by Sobolev's inequality \cite[Theorem
8.4]{LiebLoss}.

On the other hand, if $\gamma \in
\mathcal{A}$ then  
\begin{align}\label{eq:Daub}
  \rho_{\gamma} \in L^1(\mathbb{R}^3)
  \cap L^{4/3}(\mathbb{R}^3)\,.
\end{align} 
This follows from Daubechies' inequality, see \cite[pp.\
519--520]{D}. By H\"older's inquality, $\rho_{\gamma }\in
L^{6/5}(\rz^3)$. The Hardy-Littlewood-Sobolev inequality \cite[Theorem
4.3]{LiebLoss} then implies that 
$\mathcal{D}(\gamma)$ (see \eqref{def:DirCoul}) is finite. Finally, 
$\mathcal{E}x(\gamma)\leq\mathcal{D}(\gamma)$, since
\begin{align*}
   \mathcal{D}(\gamma)&-\mathcal{E}x(\gamma)\\
   &=\frac12\sum_{i,j}\lambda_i\lambda_j\sum_{\sigma,\tau=1}^q
   \int_{\rz^3}\int_{\rz^3}
   \frac{|u_i(\bx,\sigma)u_j(\by,\tau)
   -u_j(\bx,\sigma)u_i(\by,\tau)|^2}
   {|\bx-\by|}\,d\bx d\by\ge 0\,.
\end{align*}

Therefore, \(\mathcal{E}^{\rm HF}\) defined by
\eqref{eq:HF-functional} extends to \(\gamma\in\mathcal{A}\). 
This way, with \(h_0\) defined as in \eqref{1p},
\begin{equation*}
  \Tr[h_0 \gamma]
  = \Tr[E(\textbf{p})\gamma]
 -\alpha^{-1}\Tr[\gamma]-\Tr[V\gamma]\,,
\end{equation*}
and so
\begin{align}\label{def:eHFconvex}
    \mathcal{E}^{\rm HF}(\gamma)=
   \alpha^{-1}\Tr[h_0\gamma]+\mathcal{D} 
  (\gamma)-\mathcal{E}x(\gamma)\,,\  \gamma\in\mathcal{A}\,.
\end{align}

Consider
$\gamma\in\mathcal{A}$ and define, with \(\rho_\gamma\) as in
\eqref{def:rho}, 
\begin{equation}\label{Rgamma}
  R_{\gamma}(\textbf{x}):=\int_{\mathbb{R}^3}
  \frac{\rho_{\gamma}(\textbf{y})}{|\textbf{x}-\textbf{y}|}\,d
  \textbf{y}\,.
\end{equation}
We have that
\begin{align}\label{eq:propR-gamma}
  R_{\gamma}\in L^{\infty}(\rz^3)\cap L^3(\rz^3)\,. 
\end{align}
This follows from \eqref{Kato's ineq} (for \(L^\infty\)),
and \eqref{eq:Daub}
and the weak Young inequality \cite[p.\ 107]{LiebLoss}
(for \(L^3\)). Next, define the operator $K_{\gamma}$ with integral kernel
\begin{equation}\label{def:K-gamma}
  K_{\gamma}(\bx,\sigma; \by,\tau):=\frac{\gamma(\bx,\sigma;
  \by,\tau)}{|\bx-\by|}\,.
\end{equation}
The operator \(K_\gamma\) is Hilbert-Schmidt; we prove this fact in
Lemma~\ref{spectrum} below.

Note that, using \eqref{ker} and the Cauchy-Schwarz inequality,
\((u,R_{\gamma}u)\ge (u,K_\gamma u)\) (multiplication by
\(R_{\gamma}\) is in each component). Denote by
$\mathfrak{b}_{\gamma}$ the (non-negative) quadratic form given by 
\begin{equation*}
  \mathfrak{b}_{\gamma}[u,v]:=\alpha (u, R_{\gamma}v)-\alpha (u,K_{\gamma}v)
  \ \text{ for }\ u, v \in H^{1/2}(\mathbb{R}^3;\mathbb{C}^q)\,.
\end{equation*}
Then, using \((u,K_\gamma u)\geq0\) and \eqref{Kato's ineq}, 
\begin{align*}
  0\leq\mathfrak{b}_{\gamma}[u,u]\leq\alpha (u,R_{\gamma}u)  
  =\alpha\sum_{\sigma=1}^q\int_{\rz^3}\int_{\rz^3}
 \frac{\rho_{\gamma}(\by)|u(\bx,\sigma)|^2}{|\bx-\by|}
  \,d\bx\,d\by
  \leq  \alpha \frac{2}{\pi} \Tr[\gamma]\,\mathfrak{e}[u,u]\,.
\end{align*}
Therefore (by the statements and proofs of
\cite[Theorem~X.17]{RS2} and \cite[Theorem~VIII.15]{RS1}), there
exists a unique self-adjoint 
operator $h_{\gamma}$ (called the {\it Hartree-Fock operator
  associated to 
\(\gamma\)}), which is bounded below (by \({}-\alpha^{-1}\)), with
quadratic form domain $H^{1/2}(\mathbb{R}^3; \mathbb{C}^q)$ and such
that 
\begin{equation}\label{hgamma}
  (u, h_{\gamma} v)= \mathfrak{t}[u,v]-\mathfrak{v}[u,v]+ 
  \mathfrak{b}_{\gamma}[u,v]\ \text{ for }\ u,v\in
  H^{1/2}(\mathbb{R}^3;\mathbb{C}^q)\,. 
\end{equation}
The operator \(h_\gamma\) has infinitely many eigenvalues in
\([-\alpha^{-1},0)\) (when \(N<Z\)), and \(\sigma_{\rm
  ess}(h_{\gamma})=[0,\infty)\); 
both of these facts will be proved in Lemma~\ref{spectrum} below.

The main result of this paper is the following theorem.
\begin{theorem}\label{HF2}
Let $Z\alpha <2/\pi $, and let $N\ge2$ be a positive integer such
that $N<Z+1$. 

Then there exists an $N$-dimensional projection 
$\gamma^{\rm HF}=\gamma^{\rm HF}(N,Z,\alpha)$ minimizing the Hartree-Fock energy
functional \(\mathcal{E}^{\rm HF}\) given by \eqref{eq:HF-functional}, that is,
\(E^{\rm HF}(N,Z,\alpha)\) 
in \eqref{eq:HF-energy2}
(and therefore, in \eqref{eq:HF-energy})
is attained. In fact,
\begin{align}\label{eq:allSame}  \nonumber
  \mathcal{E}^{\rm HF}(\gamma^{\rm HF})=E^{\rm HF}(N,Z,\alpha)
 &=\inf\big\{\mathcal{E}^{\rm HF}(\gamma
  )\,\big|\,\gamma \in \mathcal{A}, \gamma^{2}=\gamma, 
 \Tr[\gamma]=N\big\}
 \\&=\inf\big\{\mathcal{E}^{\rm HF}(\gamma)\,\big|\,\gamma \in
  \mathcal{A},~ \Tr[\gamma]=N\}
  \nonumber
  \\&=\inf\big\{\mathcal{E}^{\rm HF}(\gamma)\,\big|\,\gamma \in
  \mathcal{A},~ \Tr[\gamma]\le N\}\,.
\end{align}

Moreover, one can write 
\begin{align}\label{eq:formMinimizer}
  \gamma ^{\rm HF}(\bx,\sigma;\by,\tau)
  =\sum_{i=1}^{N}\varphi_{i}(\bx,\sigma)
  \overline{\varphi_{i}(\by,\tau)}\,, 
\end{align}
with \(\varphi_i\in H^{1/2}(\rz^3;\cz^q),i=1,\ldots,N\), ortnonormal,
such that the {\rm Hartree-Fock orbitals} \(\{\varphi_i\}_{i=1}^N\)
satisfy: 
\begin{itemize}
\item[(i)] With
 $h_{\gamma^{\rm HF}}$ as defined in  
 \eqref{hgamma},
 \begin{align}\label{eq:HF-equations}
   h_{\gamma^{\rm HF}}\varphi_{i}=\varepsilon _{i}\varphi_{i}\ ,\ i=1,\ldots,N\,, 
 \end{align}
 with $0>\varepsilon_N\geq\ldots\geq\varepsilon_{1}>{}-\alpha^{-1}$
the \(N\) lowest eigenvalues of \(h_{\gamma^{\rm HF}}\). 
 \item[(ii)] For \(i=1,\ldots,N\),
   \begin{align}
     \label{eq:regularity}
      \varphi_i\in C^{\infty}(\rz^3\setminus\{0\};\cz^q)\,.
   \end{align}
 \item[(iii)] For all \(R>0\) and
   \(\beta<\nu_{\varepsilon_N}:=\sqrt{-\varepsilon_N(2\alpha^{-1}+\varepsilon_N)}\), 
   there exists \(C=C(R,\beta)>0\) such that for 
   \(i=1,\ldots,N\), 
   \begin{align}
     \label{eq:decay}
       |\varphi_i(\bx)|\le C\,e^{-\beta|\bx|}\quad \text{
         for } \quad |\bx|\geq R\,.
   \end{align}
\end{itemize}
\end{theorem}
\begin{remark}
 \(\, \)
\begin{enumerate}
\item[\rm (i)] In fact, we prove that \eqref{eq:regularity} holds for
  {\emph{any}} eigenfunction \(\varphi\) of 
  \(h_{\gamma^{\rm HF}}\), and  \eqref{eq:decay} for those
  corresponding to {\it negative} eigenvalues \(\varepsilon\).
 More precisely, if \(h_{\gamma^{\rm HF}}\varphi=\varepsilon\varphi\)
 for some \(\varepsilon\in[\varepsilon_N,0)\), then \eqref{eq:decay}
 holds for \(\varphi\) for all
 \(\beta<\nu_{\varepsilon}:=\sqrt{-\varepsilon(2\alpha^{-1}+\varepsilon)}\) 
 for some \(C=C(R,\beta)>0\).  
\item[\rm (ii)] Note that, in general, eigenfunctions of
  \(h_{\gamma^{\rm HF}}\) can be unbounded at \(\bx=0\); therefore
  \eqref{eq:regularity} and \eqref{eq:decay} can only be expected to
  hold away from the origin.
\item[\rm (iii)] Both the regularity and the exponential decay above are
  similar to the results in the non-relativistic case (i.e., for the operator in
  \eqref{eq:non-rel}; see \cite{LiebSimonHF}). However, the proof of
  Theorem~\ref{HF2} is considerably more complicated due to, on one
  hand, the non-locality of the kinetic energy operator \(E(\bp)\),
  and, on the other hand, the fact that the Hartree-Fock operator
  \(h_{\gamma^{\rm HF}}\) is only given as a {\emph{form}} sum for
  \(Z\alpha\in[1/2,2/\pi)\). 
\item[\rm (iv)] We show the existence of the Hartree-Fock minimizer by
  solving the minimization problem on the set of density
  matrices. This method was introduced in \cite{ReSol}. The same
  method was used in \cite{BFHS} in the Dirac-Fock case.
\item[\rm (v)] As mentioned earlier, we have to assume that
  \(Z\alpha<2/\pi\); the reason is that our proof that
  \(\Tr[E(\bp)\gamma_n]\) is uniformly bounded for a minimizing
  sequence \(\{\gamma_n\}_{n\in\nz}\) does not work in the critical
case \(Z\alpha=2/\pi\). 
\item[\rm (vi)] For simplicity of notation, we give the proof of
  Theorem~\ref{HF2} only in the spinless case. It will be obvious that
  the proof also works in the general case.
\item[\rm (vii)] As will be clear from the proofs, the
  statements of Theorem~\ref{HF2} (appropriately modified) also hold
  for molecules. More explicitely, for a molecule with \(K\) nuclei of
  charges \(Z_1,\ldots,Z_K\), fixed at \(R_1,\ldots,R_K\in\rz^3\),
  replace \(\mathfrak{v}\) in \eqref{b} by
  \begin{align}\label{V-molecule}
    \mathfrak{v}[u,v]:=\sum_{k=1}^{K}(V_k^{1/2}u, V_k^{1/2} v)\ \text{ for }\ u,v \in
  H^{1/2}(\mathbb{R}^3;\mathbb{C}^q)\,,
  \end{align}
  with \(V_k(\bx)=Z_k\alpha/|\bx-R_k|, Z_k\alpha<2/\pi\). Then, for
  \(N<1+\sum_{k=1}^KZ_k\),  there exists a Hartree-Fock
  minimizer, and the corresponding Hartree-Fock orbitals have the
  regularity and decay properties 
  as stated in  Theorem~\ref{HF2}, away from each nucleus.

\end{enumerate}
\end{remark}
\section{Proof of Theorem~\ref{HF2}}
\label{section:proof}
\subsection{Existence of the Hartree-Fock minimizer} 
The proof of the
existence of an \(N\)-dimensional projection \(\gamma^{\rm HF}\)
minimizing \(\mathcal{E}^{\rm HF}\), the
equalities in \eqref{eq:allSame}, and that the corresponding
Hartree-Fock orbitals \(\{\varphi_i\}_{i=1}^N\) solve the Hartree-Fock
equations \eqref{eq:HF-equations}, will be a consequence of the
following two lemmas.
\begin{lemma}\label{existence}
Let $Z\alpha <2/\pi $ and $N\in\nz$. Then
\begin{align*}
  E^{\rm HF}_{\leq}(N,Z,\alpha):=\inf\big\{\mathcal{E}^{\rm
    HF}(\gamma)\,\big|\,\gamma\in 
  \mathcal{A}, \Tr[\gamma]\leq N\big\}
\end{align*}
is attained.
\end{lemma}

\begin{lemma}\label{spectrum}
Let \(\gamma\in\mathcal{A}\).
Then the operator \(K_\gamma\), defined by \eqref{def:K-gamma}, is
Hilbert-Schmidt. If  \(Z\alpha<2/\pi\) then the operator $h_{\gamma}$,
defined in \eqref{hgamma}, satisfies \(\sigma_{\rm ess}(h_{\gamma})=[0,\infty)\).
If furthermore \(\Tr[\gamma]<Z\), 
then \(h_{\gamma}\) has infinitely many eigenvalues in $[-\alpha
^{-1},0)$.
\end{lemma}

Before proving these two lemmas, we use them to prove the parts of
Theorem~\ref{HF2} mentioned above. 
\begin{proof}
For computational reasons we first state and prove a lemma in the
spirit of \cite[Lemma 1]{B}.
\begin{lemma}\label{bach00}
Let $\gamma\in \mathcal{A}$, $u_1,u_2\in H^{1/2}(\R^3)$, and 
let $\epsilon_1,\epsilon_2\in \R$ be such that \(\tilde\gamma\)
given by
\begin{align}\label{eq:rank2}
  \tilde\gamma(\bx,\by)&:=\gamma(\bx,\by)
  +\gamma_{u}(\bx,\by)\,,
  \\
  \label{eq:rank2bis}
  \gamma_{u}(\bx,\by)&:=\gamma_{u_1,u_2}(\bx,\by)=\epsilon_1
  u_1(\bx)\overline{u_1(\by)}
  +\epsilon_2 u_2(\bx)\overline{u_2(\by)}
\end{align}
is again an element of $\mathcal{A}$.

Then we have that
\begin{equation}\label{bach11}
  \mathcal{E}^{\rm HF}(\tilde\gamma)=\mathcal{E}^{\rm HF}(\gamma)+
  \alpha^{-1}\epsilon_1(u_1,h_\gamma
  u_1)+\alpha^{-1}\epsilon_2(u_2,h_\gamma u_2)+ 
  \epsilon_1\epsilon_2 R_{u}\,,
\end{equation}
where \(h_{\gamma}\) is given in \eqref{hgamma}, and 
\begin{equation}\label{bach11.1}
  R_{u}:=R_{u_1,u_2}=\frac12\int_{\R^3}\int_{\R^3}\frac{|u_1(\bx)u_2(\by)
  -u_2(\bx)u_1(\by)|^2}{|\bx-\by|}\,d\bx d\by\,. 
\end{equation}
\end{lemma} 
\begin{pf*}{\it Proof of Lemma~\ref{bach00}}
We have that
\begin{align}\label{bach12}\nonumber
  \mathcal{E}^{\rm HF}(\tilde\gamma) =&\; \mathcal{E}^{\rm
    HF}(\gamma)+ \alpha^{-1}{\rm Tr}[h_0\gamma_{u}] +\int_{\R^3}\int_{\R^3}
  \frac{\rho_\gamma(\bx)\rho_{\gamma_u}(\by)}{|\bx-\by|}\,d\bx d\by\\\nonumber
  &-\int_{\R^3}\int_{\R^3}\frac{\gamma(\bx,\by)
  \overline{\gamma_{u}(\bx,\by)}}{|\bx-\by|}\,d\bx 
  d\by
  +\frac12\int_{\R^3}\int_{\R^3}\frac{\rho_{\gamma_{u}}(\bx)
  \rho_{\gamma_{u}}(\by)}{|\bx-\by|}\,d\bx
  d\by\\\nonumber
  &-\frac12\int_{\R^3}\int_{\R^3}\frac{\gamma_{u}(\bx,\by)
  \overline{\gamma_{u}(\bx,\by)}}
  {|\bx-\by|}\,d\bx d\by\\
  =& \;\mathcal{E}^{\rm HF}(\gamma)+\alpha^{-1}\epsilon_1(u_1,h_\gamma
  u_1)+\alpha^{-1}\epsilon_2(u_2,h_\gamma
  u_2)\\
  &+\frac12\int_{\R^3}\int_{\R^3}\frac{\rho_{\gamma_{u}}(\bx)
  \rho_{\gamma_{u}}(\by)}{|\bx-\by|}\,d\bx
  d\by-\frac12\int_{\R^3}\int_{\R^3}\frac{\gamma_{u}(\bx,\by)
  \overline{\gamma_{u}(\bx,\by)}}
  {|\bx-\by|}\,d\bx d\by\,.\nonumber
\end{align}
Using \eqref{eq:rank2bis},
that
$\rho_{\gamma_u}(\bx)=\epsilon_1|u_1(\bx)|^2+\epsilon_2|u_2(\bx)|^2$,
and \eqref{bach11.1}, we obtain \eqref{bach11}.
\end{pf*}
By Lemma \ref{existence} a minimizer $\gamma ^{\rm HF}\in\mathcal{A}$, with
$\Tr[\gamma^{\rm HF}]\leq N,$ exists. We may 
write  
\begin{align}\label{eq:expGamma}
  \gamma ^{\rm HF}(\mathbf{x},\by)
  =\sum_{k}\lambda_{k}\varphi_{k}(\mathbf{x})\overline{\varphi_{k}(\by)}\,,  
\end{align}
with $1\geq \lambda _{1}\geq \dots \geq 0$ and
\(\{\varphi_k\}_{k}\subset H^{1/2}(\rz^3)\) an orthonormal (in
\(L^2(\rz^3)\)) system (it might be finite). Extend
\(\{\varphi_k\}_{k}\) to an orthonormal basis
\(\{\varphi_k\}_{k}\cup\{u_{\ell}\}_{\ell\in\nz}\) for
\(L^2(\rz^3)\), with  \(u_\ell\in
H^{1/2}(\rz^3)\).

Let $K+1$ be the first index
such that $\lambda_{K+1}<1$. 
Fix \(j\in\{1,\ldots,K\}\), choose $u\in\{\varphi_k\}_{k\ge
  K+1}\cup\{u_\ell\}_{\ell\in\nz}\), and consider, for  
\(\epsilon\) to be chosen, 
\begin{align*}
  \gamma _{\epsilon }^{(j)}(\bx,\by):=\sum_{k\neq j}\lambda_{k}
  \varphi_{k}(\bx)\varphi_{k}^{\ast}(\by)+\frac{1}{1+m\epsilon^{2}}
  \big(\varphi_{j}(\bx)+\epsilon u(\bx)\big)
  \big(\,\overline{\varphi_{j}(\by)}+\epsilon \overline{u(\by)}\,\big)\,. 
\end{align*}
Choosing \(m\geq1\) assures that $\Tr[\gamma_{\epsilon}^{(j)}]\leq
N$. Then $0\leq\gamma_{\epsilon}^{(j)}\leq
\mbox{\textrm{Id}}$ for $|\epsilon|$ small enough (depending on
\(u\)). Since \(\gamma^{\rm HF}\) minimizes \(\mathcal{E}^{\rm HF}\),
and \(\gamma_{0}^{(j)}=\gamma^{\rm HF}\), 
\begin{align*}
  0=\left.\frac{d}{d\epsilon}(\mathcal{E}^{\rm HF})
  \big(\gamma_{\epsilon}^{(j)}\big)
  \right\vert_{\epsilon=0}
  =\alpha^{-1}(\varphi_{j},h_{\gamma^{\rm HF}}u)
  +\alpha^{-1}(u,h_{\gamma^{\rm HF}}\varphi_{j})\,. 
\end{align*}
Repeating the computation for \(iu\) we get
that $(u,h_{\gamma^{\rm HF}}\varphi_{j})=0$, from which it
follows that $h_{\gamma^{\rm HF}}$ maps 
$\Span\{\varphi_{1},\dots,\varphi_{K}\}$ into itself. Diagonalising
the restriction of \(h_{\gamma^{\rm HF}}\) to
  $\Span\{\varphi_{1},\dots,\varphi_{K}\}$, 
we can choose 
  \(\varphi_1,\ldots,\varphi_K\) to be eigenfunctions of
  \(h_{\gamma^{\rm HF}}\) with eigenvalues
  \(\varepsilon_{n_1},\ldots,\varepsilon_{n_K}\), \(n_j\in\mathbb{N}\)
  (numbering the 
  eigenvalues of \(h_{\gamma^{\rm HF}}\) in increasing order,
  \(-\alpha^{-1}<\varepsilon_{1}\le\varepsilon_{2}\le\cdots\)). 
  Since \(\lambda_1=\cdots=\lambda_K=1\),
  this does not change \eqref{eq:expGamma}.

To show that, for
$j>K$, $\varphi_{j}$ is also an eigenfunction of $h_{\gamma^{\rm HF}}$
(corresponding to an eigenvalue \(\varepsilon_{n_j}\))
one repeats the argument above, with $u\in\{\varphi_k\}_{k\neq
  1,\ldots,K,j}\cup\{u_\ell\}_{\ell\in\nz}\), and
\begin{align*}
  \gamma _{\epsilon }^{(j)}(\bx,\by)=\sum_{k\neq j}\lambda_{k}
  \varphi_{k}(\bx)\overline{\varphi_{k}(\by)}+\frac{\lambda_j}{1+m\epsilon^{2}} 
  \big(\varphi_{j}(\bx)+\epsilon u(\bx)\big)
  \big(\,\overline{\varphi_{j}(\by)}+\epsilon \overline{u(\by)}\,\big)\,. 
\end{align*}

Moreover, the eigenvalues $\varepsilon_{n_k}$ (of  $h_{\gamma^{\rm HF}}$) 
corresponding to the eigenfunctions $\varphi_k$ are non-positive.
In fact, if $\varepsilon_{n_k}>0$, then we could lower the
energy: Define $\tilde\gamma(\bx,\by)=\gamma^{\rm
  HF}(\bx,\by)-\lambda_k 
\varphi_k(\bx)\overline{\varphi_k(\by)}$, then, using 
Lemma~\ref{bach00}, we get that $\mathcal{E}^{\rm
  HF}(\tilde\gamma)=\mathcal{E}^{\rm
  HF}(\gamma^{\rm HF})-\alpha^{-1}\lambda_k \varepsilon_{n_k}<\mathcal{E}^{\rm
  HF}(\gamma^{\rm HF}).$

It remains to show that $\Tr[\gamma^{\rm HF}]=N$, that 
$\gamma^{\rm HF}$ is a projection, and that the
$\{\varphi_j\}_{j=1}^{N}$ are eigenfunctions  corresponding
to the {\it lowest} (negative) eigenvalues of $h_{\gamma^{\rm HF}}$ (that is, 
to \(\varepsilon_{1}\le\varepsilon_{2}\le\cdots\le\varepsilon_{N}<0\)).

Consider first the case \(N<Z\).  Assume, for contradiction, that
$\Tr[\gamma^{\rm HF}]<N$.  Let
\(K\in\N\) be the multiplicity of the eigenvalue $1$ in
\eqref{eq:expGamma}. 
Since (by Lemma~\ref{spectrum}), for $N<Z$, $h_{\gamma^{\rm HF}}$ has infinitely many
eigenvalues in \([-\alpha^{-1},0)\) we can
find a (normalized) eigenfunction $u$, corresponding to a negative
eigenvalue of \(h_{\gamma^{\rm HF}}\), and orthogonal to $\varphi_1,
\dots, \varphi_K$. Let $\epsilon>0$ be sufficiently small that
$\gamma(\bx,\by):=\gamma^{\rm HF}(\bx,\by)
+\epsilon u(\bx)\overline{u(\by)}$ defines a density matrix
satisfying $\Tr[\gamma]\leq N$. 
By Lemma \ref{bach00} (with $u_1=u,\epsilon_1=\epsilon$ and
$\epsilon_2=0$) we get that \begin{equation}\label{eq:lowerE} 
\mathcal{E}^{\rm
  HF}(\gamma)=\mathcal{E}^{\rm HF}(\gamma^{\rm HF}) +\epsilon
\alpha^{-1}(u,h_{\gamma^{\rm HF}}u) <\mathcal{E}^{\rm HF}(\gamma^{\rm
  HF})\,,
  \end{equation}
leading to a contradiction. Hence,
$\Tr[\gamma^{\rm HF}]=N$. That $\gamma^{\rm HF}$ is a projection follows
from Lieb's Variational Principle (see \cite{Lieb}) which we prove for
completeness. If this is not the case, there exist indices $p,q$ such
that $0 < \lambda_p, \lambda_q<1$. Consider $\tilde{\gamma}(\bx,\by):=
\gamma^{\rm HF}(\bx,\by)+ \epsilon \varphi_q(\bx)
\overline{\varphi_q(\by)} - \epsilon \varphi_p(\bx)
\overline{\varphi_p(\by)}$ with $\epsilon$ such that $0 \leq
\tilde{\gamma} \leq \mbox{\textrm{Id}}$. Choose $\epsilon >0 $ if
$\varepsilon_{n_q} \leq \varepsilon_{n_p}$ and $\epsilon<0$
otherwise. By Lemma \ref{bach00}, we get that
$\mathcal{E}^{\rm HF}(\tilde{\gamma})< \mathcal{E}^{\rm
  HF}(\gamma^{\rm HF})$. 

Consider now the case \(Z\leq N<Z+1\) (and \(N\ge2\)), so that \(N-1<Z\). Let 
\(\gamma^{\rm HF}_{N-1}\) denote the density matrix where
\begin{equation*}
   \inf\big\{\mathcal{E}^{\rm HF}(\gamma)\,\big|\,\gamma \in
  \mathcal{A},~ \Tr[\gamma]\le N-1\}
\end{equation*}
is attained. By the above, $\Tr[\gamma^{\rm HF}_{N-1}]=N-1$ and
$\gamma^{\rm HF}_{N-1}$ is a projection, so its
integral kernel is given by
\begin{equation*}
  \gamma^{\rm HF}_{N-1}(\bx,\by)= \sum_{i=1}^{N-1} \phi_i(\bx)
  \overline{\phi_i(\by)}\,, 
\end{equation*}
where the $\phi_i$'s are eigenfunctions of $h_{\gamma^{\rm HF}_{N-1}}$.

We first prove that 
\begin{equation}\label{N+1}
   \inf\big\{\mathcal{E}^{\rm HF}(\gamma)\,\big|\,\gamma \in
  \mathcal{A},~ \Tr[\gamma]\le N\}
\end{equation}
is not attained at the density matrix $\gamma^{\rm HF}_{N-1} $ 
by constructing a density matrix $\tilde{\gamma}$ with
$\Tr[\tilde{\gamma}]\leq N$ such that
$\mathcal{E}^{\rm HF}(\tilde{\gamma})<
\mathcal{E}^{\rm HF}(\gamma^{\rm HF}_{N-1})$. 
Indeed, since $h_{\gamma^{\rm HF}_{N-1}}$
has infinitely many strictly negative eigenvalues (by Lemma~\ref{spectrum};
\(N-1<Z\)) there exists a 
(normalized) eigenfunction $u$ of $h_{\gamma^{\rm HF}_{N-1}}$ corresponding to
a negative eigenvalue, and orthogonal to $\Span\{\phi_1, \dots,
\phi_{N-1}\}$. Let $\tilde{\gamma}$ be defined by
\begin{equation*}
   \tilde{\gamma}(\bx,\by)=\gamma^{\rm HF}_{N-1}(\bx,\by) 
   + u(\bx) \overline{u(\by)}\,.
\end{equation*}
Then $\Tr[\tilde{\gamma}]=N$ and, by a computation like in
\eqref{eq:lowerE}, 
\begin{equation*}
   \mathcal{E}^{\rm HF}(\tilde{\gamma})= \mathcal{E}^{\rm
     HF}(\gamma^{\rm HF}_{N-1})+ \alpha^{-1}(u, h_{\gamma^{\rm HF}_{N-1}} u) <
   \mathcal{E}^{\rm HF}(\gamma^{\rm HF}_{N-1})\,.
\end{equation*}
Hence,
\begin{align}\label{inf-N-1&N}
  \inf\big\{\mathcal{E}^{\rm HF}(\gamma)\,\big|\,\gamma\in\mathcal{A},
  \Tr[\gamma]\le N\big\}
  <\inf\big\{\mathcal{E}^{\rm HF}(\gamma)\,\big|\,\gamma\in\mathcal{A},
  \Tr[\gamma]\le N-1\big\}\,.
\end{align}

Let $\gamma_{N}$ be a  density matrix where \eqref{N+1} is attained (the existence
of such a minimizer follows, as before, from Lemma~\ref{existence}). By the above
it follows that $N-1 < \Tr[\gamma_{N}]\leq N$. We now show that there
exists a minimizer \(\gamma^{\rm HF}\) with \(\Tr[\gamma^{\rm HF}]=N\).

The integral kernel of
$\gamma_{N}$ is given by 
\begin{equation*}
  \gamma_{N}(\bx,\by)= \sum_{j} \lambda_j \varphi_{j}(\bx)
  \overline{\varphi_j(\by)}\,, 
\end{equation*} 
where $1\ge\lambda_1\ge\dots\ge0$ and the $\varphi_j$'s are
(orthonormal) eigenfunctions of 
$h_{\gamma_{N}}$. If $\Tr[\gamma_{N}]<N$ we can define
a new density matrix $\tilde{\gamma}$ with $\Tr[\tilde{\gamma}]\leq
N$ and \(\mathcal{E}^{\rm HF}(\tilde\gamma)\le\mathcal{E}^{\rm
  HF}(\gamma_{N})\).  Indeed, if 
$\Tr[\gamma_{N}]<N$ (and 
bigger than $N-1$) then there exists a (first) $j_0$ such that
$0<\lambda_{j_0}<1$. We define $\tilde{\gamma}$ with integral kernel 
\begin{equation}\label{eq:trialDensity}
  \tilde{\gamma}(\bx,\by)=\gamma_{N}(\bx,\by)
  + r\varphi_{j_0}(\bx) \overline{\varphi_{j_0}(\by)}\,,
\end{equation}
with \(r=\min\{1-\lambda_{j_0},N-\Tr[\gamma_{N}]\}>0\).
Recall that \(h_{\gamma_{N}}\varphi_j=\varepsilon_{n_j}\varphi_j\),
\(\varepsilon_{n_j}\le0\),  for all
\(j\). By Lemma \ref{bach00} 
we have that
$$\mathcal{E}^{\rm HF}(\tilde{\gamma})= \mathcal{E}^{\rm
     HF}(\gamma_{N})+ \alpha^{-1}r\varepsilon_{n_{j_0}}\,.
$$
If \(\varepsilon_{n_{j_0}}<0\), 
it follows that 
$\mathcal{E}^{\rm HF}(\tilde{\gamma})<\mathcal{E}^{\rm HF}(\gamma_N)$. 
On the other hand, if \(\varepsilon_{n_{j_0}}=0\), then 
\(\mathcal{E}^{\rm HF}(\tilde{\gamma})= \mathcal{E}^{\rm
  HF}(\gamma_{N})\), and \(\Tr[\gamma_{N}]<\Tr[\tilde{\gamma}]\le N\). Either
\(\Tr[\tilde{\gamma}]=N\), in which case we let
\(\gamma^{\rm HF}:=\tilde{\gamma}\), and, as above, we are done.  Or,
we repeat all of the above argument on
\begin{equation*}
   \tilde{\gamma}(\bx,\by)= \sum_{j=1}^{j_0}\varphi_{j}(\bx)
  \overline{\varphi_j(\by)}
  +\sum_{j>j_0} \lambda_j \varphi_{j}(\bx)
  \overline{\varphi_j(\by)}\,.
\end{equation*} 
Since the trace stays bounded by \(N\), this procedure has to stop
eventually. 
Hence, with \(\gamma^{\rm HF}\) the resulting density matrix,
$\Tr[\gamma^{\rm HF}]=N$ and by Lieb's Variational Principle it 
follows (as above) that $\gamma^{\rm HF}$ is a projection. 

Finally, let $\{\varphi_j\}$ be the eigenfunctions
of $h_{\gamma^{\rm HF}}$, now numbered
corresponding to the eigenvalues
$\varepsilon_1\le\varepsilon_2\le\cdots$, where $\varepsilon_1$ is the
lowest eigenvalue of $h_{\gamma^{\rm HF}}$. We know that, for some
\(j_1,\ldots,j_N\in\mathbb{N}\),  
$$
  \gamma^{\rm HF}(\bx,\by)=\sum_{k=1}^N
  \varphi_{j_k}(\bx)\overline{\varphi_{j_k}(\by)}\,.
$$ 
Suppose for contradiction that $\{\varepsilon_{j_1},\dots,
\varepsilon_{j_N}\}\not=\{\varepsilon_1, \dots,
\varepsilon_N\}$. Then there exists a $k\in \{1,\dots,N\}$ with
$\varepsilon_{j_k}>\varepsilon_k$. For $\delta\in(0,1)$ define
$$
  \tilde\gamma(\bx,\by)=\gamma^{\rm HF}(\bx,\by) 
  +\delta\varphi_k(\bx)\overline{\varphi_k(\by)}
  -\delta\varphi_{j_k}(\bx)\overline{\varphi_{j_k}(\by)}\,.
$$
By Lemma \ref{bach00}, 
$$ 
 \mathcal{E}^{\rm HF}(\tilde{\gamma})= \mathcal{E}^{\rm
  HF}(\gamma^{\rm HF})
 +\delta\alpha^{-1}(\varepsilon_k-\varepsilon_{j_k})
 -\delta^2R_{\varphi_j,\varphi_{j_k}}<\mathcal{E}^{\rm HF}(\gamma^{\rm
   HF})\,, 
$$
where the last inequality follows by choosing $\delta$ small
enough. 

It remains to prove that \(\varepsilon_1,\ldots,\varepsilon_{N}\) are
strictly negative. For \(N<Z\) this follows directly from
Lemma~\ref{spectrum}. In the case \(Z\le N<Z+1\), assume, for
contradiction, that \(\varepsilon_{N}=0\); then the density matrix
\begin{align*}
  \tilde\gamma(\bx,\by):=\gamma^{\rm
    HF}(\bx,\by)-\varphi_{N}(\bx) \overline{\varphi_{N}(\by)}
\end{align*}
satisfies \(\mathcal{E}^{\rm HF}(\tilde\gamma)=\mathcal{E}^{\rm
  HF}(\gamma^{\rm HF})\) (by Lemma~\ref{bach00}) and
\(\Tr[\tilde\gamma]=N-1\). This is a contradiction to
\eqref{inf-N-1&N}. 

This finishes the  
proof of the first part of Theorem~\ref{HF2}.
\end{proof}

It remains to prove Lemma~\ref{existence} and Lemma~\ref{spectrum}.
\begin{pf*}{\it Proof of Lemma~\ref{existence}}
We minimize on
density matrices following the method in \cite{ReSol}. In the
pseudorelativistic context one faces the problem that the Coulomb
potential is not relatively compact with respect to the kinetic
energy. This problem has been adressed in \cite{BFHS} and we follow
the idea therein. 

The quantity $E^{\rm HF}_{\leq}(N,Z,\alpha)$ is finite since for any
density matrix $\gamma $, with $\Tr[\gamma]\leq N$,  
\begin{align*}
  \mathcal{E}^{\rm HF}(\gamma) \geq \alpha^{-1}\big\{\Tr[E
  (\bp)\gamma]- \alpha^{-1}N
  -\Tr[V\gamma] \big\}
  \geq {}-\alpha ^{-2}N\,.
\end{align*}
Here we used that $\mathcal{D}(\gamma )-\mathcal{E}x(\gamma)\geq 0$,
and \eqref{Kato's ineq} (see also \eqref{kintr} and
\eqref{eq:potntialQuad}). 

Let $\{\gamma _{n}\}_{n=1}^\infty$ be a minimizing sequence 
for $E^{\rm HF}_{\leq}(N,Z,\alpha)$, more precisely, $\gamma _{n}\in
\mathcal{A}$ (with $\mathcal{A}$ as defined in \eqref{A}),
$\Tr[\gamma_{n}] \leq N$, and
$\mathcal{E}^{\rm HF}(\gamma _{n})\leq E^{\rm HF}_{\leq}(N,Z,\alpha)+1/n$. 

The sequence $\Tr[E(\bp)\gamma_{n}]$
is uniformly bounded. Indeed, 
for every $n\in \mathbb{N}$, using \eqref{Kato's ineq}, 
\begin{align*}
  E^{\rm HF}(N,Z,\alpha)+1 &\geq
  \mathcal{E}^{\rm HF}(\gamma_{n})\geq\alpha^{-1}\big\{\Tr[E
  (\bp)\gamma_n]-\alpha^{-1}N
  -\Tr[V\gamma_n]\big\} \\ 
  &\geq \alpha^{-1}(1-Z\alpha\frac{\pi}{2})\Tr[E(\bp)\gamma_n]-\alpha^{-2}N\,. 
\end{align*}
The claim follows since \(Z\alpha<2/\pi\). It is this argument that
prevents us from proving Theorem~\ref{HF2} for the critical case
\(Z\alpha=2/\pi\). 

Define $\tilde{\gamma}_{n}:=E(\bp)^{1/2}\gamma_{n}E(\bp)^{1/2}$. Then,
by the above,
$\{\tilde{\gamma}_n\}_{n \in \mathbb{N}}$ 
is a sequence of Hilbert-Schmidt operators with uniformly bounded
Hilbert-Schmidt norm. 
Hence, by Banach-Alaoglu's theorem, there exist a
subsequence, which we denote again by $\tilde{\gamma}_{n}$, and a
Hilbert-Schmidt operator $\tilde{\gamma}_{\left(\infty \right) }$, such
that for every Hilbert-Schmidt operator $W$,
\begin{equation*}
  \Tr[W\tilde{\gamma}_{n}]\rightarrow 
  \Tr[W\tilde{\gamma}_{(\infty)}]\ , \ n\to\infty\,.
\end{equation*} 
Let $\gamma_{(\infty)}:=E(\textbf{p})^{-1/2}\tilde{\gamma}_{(\infty)}
E(\textbf{p})^{-1/2}$. We are going to show that $\gamma _{(\infty)}$
is a minimizer of \(\mathcal{E}^{\rm HF}\) (in fact, of
\(\alpha\mathcal{E}^{\rm HF}\), which is equivalent). We first prove that
$\gamma_{(\infty)}\in\mathcal{A}$, then that \(\mathcal{E}^{\rm HF}\)
is weak lower semicontinuous on \(\mathcal{A}\). 

Let $\{\psi_{k}\}_{k\in \mathbb{N}}$ be a basis of
$L^{2}(\mathbb{R}^{3})$ with $\psi_{k}\in
H^{1/2}(\mathbb{R}^{3})$. Then, for all \(k\in\nz\)\,, 
\begin{align*}
  \lim_{n\rightarrow \infty}(\psi_{k},\gamma_{n}\psi_{k}) 
  &=\lim_{n\rightarrow
    \infty}(\psi_{k},E(\textbf{p})^{-1/2}\tilde{\gamma}_{n}
  E(\textbf{p})^{-1/2}\psi_{k}) \\  
  &=(\psi_{k},\gamma _{(\infty)}\psi _{k})\,. 
\end{align*}
{F}rom this follows, by Fatou's lemma, that
\begin{equation*}
  \Tr[\gamma_{(\infty)}]=\sum_{k}(\psi_{k},\gamma
  _{(\infty)}\psi_{k})\leq 
  \liminf_{n\rightarrow \infty}\sum_{k}(\psi_{k},\gamma_{n}\psi_{k})
  =\liminf_{n\rightarrow\infty}\Tr[\gamma_{n}]\leq N\,, 
\end{equation*}
and
\begin{equation*}
  \Tr[E(\textbf{p})^{1/2}\gamma_{(\infty)}E(\textbf{p})^{1/2}]\leq
  \liminf_{n\rightarrow \infty}\Tr[E(\textbf{p})^{1/2}\gamma_{n}
  E(\textbf{p})^{1/2}]<\infty\,. 
\end{equation*}
Since also $0\leq \gamma_{(\infty)}\leq
\mbox{\textrm{Id}}$ we see that $\gamma_{(\infty)}\in\mathcal{A}$.

To reach the claim it remains to show the weak lower semicontinuity of the
functional \(\mathcal{E}^{\rm HF}\). As mentioned in the introduction,
the spectrum of the 
one-particle operator $h_0$, defined in \eqref{1p}, is discrete in
$[-\alpha^{-1},0)$ and purely absolutely continuous in
$[0,\infty)$. Let $ \Lambda_{-}(\alpha)$ denote the projection
on the 
pure point spectrum of $h_0$ and
${\Lambda}_{+}(\alpha):=\mbox{\textrm{Id}}-\Lambda_{-}(\alpha)$. We
write  
\begin{align}\label{eq:threeterms}
  \alpha \mathcal{E}^{\rm HF}(\gamma_{n})=T_{1}(\gamma_{n})
  +T_{2}(\gamma_{n})+\alpha T_{3}(\gamma_{n})\,, 
\end{align}
with
\begin{align*}
  T_{1}(\gamma_{n})&=\Tr[\Lambda_{+}(\alpha)
   h_0\Lambda_{+}(\alpha)\gamma_{n}]\ ,\quad
  T_{2}(\gamma_{n})=\Tr[\Lambda_{-}(\alpha)
  h_0\Lambda_{-}(\alpha)\gamma_{n}]\,,\\ 
  T_{3}(\gamma_{n})&=\mathcal{D}(\gamma_{n})
  -\mathcal{E}x(\gamma_{n})\,. 
\end{align*}
We consider these three terms separately.

For the first term in \eqref{eq:threeterms}, 
fix (as above) a basis $\{\psi_{k}\}_{k\in \mathbb{N}}$
  of $L^{2}(\mathbb{R}^{3})$, with   $\{\psi_{k}\}_{k\in
    \mathbb{N}}\subset H^{1/2}(\mathbb{R}^{3})$. Defining  
\begin{align*}
  f_{k}:=\big(\Lambda_{+}(\alpha)h_0
  \Lambda_{+}(\alpha)\big)^{1/2}\psi_{k}\,, 
\end{align*}
we have that
\begin{align*}
  T_{1}(\gamma_{n})
  &=\Tr\big[\big(\Lambda_{+}(\alpha)h_0 
  \Lambda_{+}(\alpha)\big)^{1/2}\gamma_{n}
  \big(\Lambda_{+}(\alpha)h_0
  \Lambda_{+}(\alpha)\big)^{1/2}\big] \\ 
   &=\sum_{k}(f_{k},\gamma_{n}f_{k}) 
   =\sum_{k}(E(\textbf{p})^{-1/2}f_{k},\tilde\gamma _{n}
   E(\bp)^{-1/2}f_{k})\,. 
\end{align*}
Since the projection
\begin{align*}
  H_{k}:=\big| E(\textbf{p})^{-1/2}f_{k}\big\rangle\,\big\langle
  E(\textbf{p})^{-1/2}f_{k}\big| 
\end{align*} is a non-negative Hilbert-Schmidt
operator, we find, by Fatou's lemma, that  
\begin{align*}
  \liminf_{n\rightarrow\infty}T_{1}(\gamma_{n})
  &=\liminf_{n\rightarrow\infty}\sum_{k}\Tr[
  H_{k}\tilde{\gamma}_{n}] 
  \geq\sum_{k}\Tr[
  H_{k}\tilde{\gamma}_{(\infty)}] =T_{1}(\gamma_{(\infty)})\,. 
\end{align*}

As for the second term in \eqref{eq:threeterms}, we have
$\lim_{n\rightarrow\infty
  }T_{2}(\gamma_{n})=T_{2}(\gamma_{(\infty)})$ since the
  operator $\Lambda_{-}(\alpha)h_0\Lambda_{-}(\alpha)$  
is Hilbert-Schmidt; see Lemma~\ref{lem:HS} in
Appendix~\ref{app:lemmata}. 

Finally, for the last term in \eqref{eq:threeterms}, 
following the reasoning in
  \cite[pp.142--143]{BFHS} (here we need that \(N\in\nz\)), we get that  
\begin{equation*}
  \liminf_{n\rightarrow \infty}T_3(\gamma_{n})
  \geq T_3(\gamma _{(\infty)})\,. 
\end{equation*}

This finishes the proof of Lemma~\ref{existence}.
\end{pf*}
\begin{pf*}{Proof of Lemma~\ref{spectrum}}
In order to prove that $K_\gamma$ is Hilbert-Schmidt it is enough to
prove that its integral kernel belongs to $L^2(\R^6)$. We have that
(see \eqref{def:K-gamma} and \eqref{ker})
\begin{align}\label{juni1}
  \int_{\R^6} |K_\gamma(\bx,\by)|^2\,d\bx d\by&=\int_{\R^6}
  \frac{|\gamma(\bx,\by)|^2}{|\bx-\by|^2}\,d\bx d\by\\
  &=\sum_{j,k}\lambda_j\lambda_k \int_{\R^6} \frac{\overline{u_k(\bx)} u_j(\bx)
    u_k(\by) \overline{u_j(\by)}}{|\bx-\by|^2}\,d\bx d\by 
  =:\sum_{j,k}\lambda_j\lambda_k I_{j,k}\,.\nonumber
\end{align}
The last integral can be estimated using the Hardy-Littlewood-Sobolev, 
H\"older, and Sobolev inequalities (in that order), to get 
\begin{equation}\label{juni2}
  I_{j,k}\le \|u_k u_j\|_{3/2}^2\le \|u_k\|_3^2 \|u_j\|_3^2\le C
  \|u_k\|_{H^{1/2}}^2\|u_j\|_{H^{1/2}}^2\,.
\end{equation}
Inserting \eqref{juni2} in \eqref{juni1} we obtain (since \(\gamma\in\mathcal{A}\))
\begin{align*}
  \int_{\R^6} |K_\gamma(\bx,\by)|^2\,d\bx d\by&\le
  C\sum_{j,k}\lambda_j\lambda_k\|u_k\|_{H^{1/2}}^2\|u_j\|_{H^{1/2}}^2= 
  C\Big(\sum_{j}\lambda_j\|u_j\|_{H^{1/2}}^2\Big)^2\\
  &=C\big(\!\Tr[E(\bp)\gamma]\big)^2<\infty\,.
\end{align*}

To prove the statement on the essential spectrum, 
define $\tilde{h}_\gamma:=h_\gamma+\alpha K_\gamma$. Since  $K_\gamma$ is
Hilbert-Schmidt, and $\sigma_{\rm ess}(h_0)=[0,\infty)$ (see the
introduction), it is 
enough to prove that 
$(\tilde{h}_\gamma+\eta)^{-1}-(h_0+\eta)^{-1}$ is compact
for some $\eta>0$ large enough \cite[Theorem XIII.14]{RS4}. Since
$\mathcal{D}(h_0)=\mathcal{D}(\tilde{h}_\gamma)
\subset\mathcal{D}(R_\gamma)$, 
we have that
\begin{equation}\label{15.06.2007-1}
  (\tilde{h}_\gamma+\eta)^{-1}-(h_0+\eta)^{-1}=
  {}-(\tilde{h}_\gamma+\eta)^{-1}\alpha R_\gamma(h_0+\eta)^{-1}\,. 
\end{equation}
From Tiktopoulos's formula (see \cite[(II.8), Section II.3]
{simon}), it follows that
\begin{align}\label{15.06.2007-2}
  &(h_0+\eta)^{-1}\nonumber\\
  &\quad=(T(\bp)+\eta)^{-1/2}
  [1-(T(\bp)+\eta)^{-1/2}V(T(\bp)+\eta)^{-1/2}]^{-1}(T(\bp)+\eta)^{-1/2}\,.
\end{align}
Since, by \eqref{eq:Kato}, $\|(T(\bp)+\eta)^{-1/2}V^{1/2}\|<1$ for
$Z\alpha<2/\pi$ and \(\eta>\alpha^{-1}\), the right side of \eqref{15.06.2007-2} 
is well defined. Inserting \eqref{15.06.2007-2} in
\eqref{15.06.2007-1} one sees that it suffices to prove that 
$R_\gamma(T(\bp)+\eta)^{-1/2}$ is compact. That this is indeed the
case follows by using \cite[Theorem XI.20]{RS3} together  
with the observation that, for $\varepsilon>0$ and
\(\eta>\alpha^{-1}\), $R_\gamma$ and $(T(\bp)+\eta)^{-1/2}$ (as a function of
$\bp$) belong to the space $L^{6+\varepsilon}(\R^3)$ (for
\(R_{\gamma}\), see \eqref{eq:propR-gamma}).  

Finally, we show that if $\Tr[\gamma]=N<Z$ then $h_{\gamma}$ has infinitely
many eigenvalues in $[-\alpha^{-1},0)$. By the min-max principle
\cite[Theorem XIII.1]{RS4} and since
$\sigma_{\rm ess}(h_{\gamma})=[0,\infty)$, it is sufficient to show that
for every $n \in \nz$ we can find $n$ orthogonal functions $u_1,
\dots, u_n$ in $L^{2}(\R^{3})$ such that $(u_i, h_{\gamma} u_i) < 0$
for $i=1, \dots,n$. 

Let $n \in \nz$. Fix $\delta:= 1-N/Z$ and let $h_{0,\delta}$ be the
unique self-adjoint operator whose quadratic form domain is
$H^{1/2}(\R^3)$ such that 
\begin{equation*}
  (u,h_{0,\delta} v) 
  =\mathfrak{t}[u,v]-\delta\,\mathfrak{v}[u,v] 
  \text{ for }u,v \in H^{1/2}(\R^3)\,. 
\end{equation*}
By \cite[Theorems 2.2 and 2.3]{Herbst},
$\sigma_{\rm ess}(h_{0,\delta})=[0,\infty)$. Moreover, $h_{0,\delta}$ has
infinitely many eigenvalues in $[-\alpha^{-1},0)$. This follows by the
min-max principle and the inequality $h_{0,\delta} \leq \alpha/2
(-\Delta)-\delta Z\alpha/|\bx|$. Hence, we can find $u_{1}, \dots,
u_{n}$ spherically symmetric and orthonormal such that $(u_{i},
h_{0,\delta} u_{i}) <0$ for $i=1, \dots, n$. Then, by the positivity of
$K_{\gamma}$, by Newton's Theorem \cite[p.\ 249]{LiebLoss}, and
since $\Tr[\gamma]=N$ we get, for $i=1,\dots,n$, that 
\begin{align*}
  (u_i , h_{\gamma} u_i) & \leq  {\mathfrak t}[u_i,u_i]-{\mathfrak
    v}[u_i,u_i]+ \alpha (u_i, R_{\gamma} u_i) \\ 
  &\leq   {\mathfrak t}[u_i,u_i]-{\mathfrak v}[u_i,u_i]
  +\frac{N}{Z}\,{\mathfrak v}[u_i,u_i] 
  = (u_i,h_{0,\delta} u_i )<0\,. 
\end{align*}
The claim follows.
\end{pf*}
\subsection{Regularity of the Hartree-Fock orbitals}
Here we prove that any eigenfunction of \(h_{\gamma^{\rm HF}}\) is in
\(C^{\infty}(\rz^3\setminus\{0\})\). 
\begin{proof}
Let $\varphi$ be a solution of $h_{\gamma^{\rm HF}}\varphi=
\varepsilon\varphi$ for some  $\varepsilon\in\rz$. Then $\varphi$
belongs to the domain of the operator and in particular to
$H^{1/2}(\rz^3;\cz^q)$. We are going to prove that
$\varphi \in H^k(\Omega)$ for all bounded smooth $\Omega \subset
\mathbb{R}^3\setminus \{0\}$ and all $k \in \mathbb{N}$. The claim
will then follow from the Sobolev imbedding theorem \cite[Theorem
4.12]{Sobolev}. 
We will use results on pseudodifferential operators; see
Appendix~\ref{app:pseudo}. We briefly summarize these here.
\begin{enumerate}
\item[1)] For all \(k,\ell\in\rz\), \(E(\bp)^{\ell}\) maps
  \(H^k(\rz^3)\) to \(H^{k-\ell}(\rz^3)\).
\item[2)] For all \(k,\ell\in\rz\), and any \(\chi\in
  C_0^\infty(\rz^3)\), the commutator \([\chi,E(\bp)^{\ell}]\)  maps
  \(H^k(\rz^3)\) to \(H^{k-\ell+1}(\rz^3)\).
\item[3)] For all \(k,\ell,m\in\rz\) and \(\chi_1, \chi_2\in
  C_{0}^{\infty}(\rz)\) with \(\supp\,\chi_1\cap\supp\,\chi_2=\emptyset\),
  \(\chi_1E(\bp)^{\ell}\chi_2\) maps
  \(H^k(\rz^3)\) to \(H^{m}(\rz^3)\). Such an operator is
  called `smoothing'.
\end{enumerate}

Fix $\Omega$ a bounded smooth subset of $\mathbb{R}^3 \setminus \{0\}$. We
proceed by induction on $k\in\nz$. Assume that $\varphi\in H^k(\Omega)$
for some $k \geq 0$, i.e.,  $\chi \varphi \in
H^k(\mathbb{R}^3)$ for all $\chi \in
C^{\infty}_0(\Omega)$. Notice that
$H^k(\mathbb{R}^3)=D(E(\textbf{p})^k)$. 

Since $\chi \varphi \in H^{k+1}(\mathbb{R}^3)$ is equivalent to
$\chi \varphi \in D(E(\textbf{p})^{k+1})$, and
$D(E(\textbf{p})^{k+1})=D((E(\textbf{p})^{k+1})^{*})$, it is sufficient
to prove that $\chi \varphi \in D((E(\textbf{p})^{k+1})^{*})$, or
equivalently, that there exists $v\in L^2(\mathbb{R}^3)$ such that
\begin{align*}
  (\chi \varphi, E(\textbf{p})^{k+1}f)=(v,f)\ \text{ for
all } f\in H^{k+1}(\mathbb{R}^3)\,. 
\end{align*}

Let $f \in H^{k+1}(\mathbb{R}^3)$. Then
\begin{align}
  \nonumber
  (\chi \varphi, E(\textbf{p})^{k+1}f)&= \mathfrak{e}(\varphi,
  E(\textbf{p})^{-1} \chi E(\textbf{p})^{k+1} f)\\ 
  \label{terms}
  &=(\varepsilon+\alpha^{-1}) (\varphi, E(\textbf{p})^{-1}\chi
  E(\textbf{p})^{k+1}f) + \mathfrak{v}(\varphi, E(\textbf{p})^{-1}\chi
  E(\textbf{p})^{k+1}f) 
  \nonumber\\
  &\quad-\mathfrak{b}_{\gamma^{\rm HF}}(\varphi, E(\textbf{p})^{-1}\chi
  E(\textbf{p})^{k+1}f)\,, 
\end{align}
where we use that $h_{\gamma^{\rm HF}} \varphi = \varepsilon
\varphi$. We study the terms in \eqref{terms} separately. In the
following, $\tilde{\chi}$ denotes a function in
$C^{\infty}_{0}(\Omega)$ with $\tilde{\chi}\equiv 1$
on $\supp\,\chi$. 

For the first term in \eqref{terms} we find that
\begin{align}\nonumber
  &(\varphi, E(\bp)^{-1}\chi E(\bp)^{k+1}f) = (\chi
  E(\bp)^{-1}\varphi, E(\bp)^{k+1}f)\\ 
  \label{S1} 
  &\quad=([\chi,E(\bp)^{-1}]\varphi,E(\bp)^{k+1}f)+
  (E(\bp)^{-1}\chi\varphi,E(\bp)^{k+1}f)\,. 
\end{align}
Since $\chi \varphi \in H^k(\mathbb{R}^3)$ by the induction
hypothesis, we have that
$E(\textbf{p})^{-1} \chi \varphi \in H^{k+1}(\mathbb{R}^3)$ and
hence there exists $w_1 \in L^2(\mathbb{R}^3)$ such that 
\begin{equation*}
  (E(\bp)^{-1}\chi \varphi, E(\bp)^{k+1}f)=(w_1,f)\,.
\end{equation*}
It remains to study the first term in \eqref{S1}. We have that
\begin{align*}
  & ([\chi,E(\bp)^{-1}]\varphi,E(\bp)^{k+1}f)\\
  &=([\chi,E(\bp)^{-1}]\tilde{\chi}\varphi,
  E(\bp)^{k+1}f)+([\chi,E(\bp)^{-1}](1-\tilde{\chi})
  \varphi,E(\bp)^{k+1}f)\,. 
\end{align*}
Since $\tilde{\chi}\varphi \in H^k(\mathbb{R}^3)$ by the induction
hypothesis, it follows from
Proposition~\ref{corcom}
that $[\chi,E(\textbf{p})^{-1}]\tilde{\chi}
\varphi$ belongs to $H^{k+2}(\mathbb{R}^3)$. On the other hand since
the supports of $\chi$ and $\tilde{\chi}$ are disjoint the operator
$[\chi,E(\textbf{p})^{-1}](1-\tilde{\chi})$ is a smoothing operator.
Hence there exists a
$w_2 \in L^2(\mathbb{R}^3)$ such that 
\begin{align*}
  ([\chi, E(\bp)^{-1}]\varphi,E(\bp)^{k+1}f)=(w_2,f)\,.
\end{align*}

As for the second term in \eqref{terms}, we find, with $\tilde{\chi}$
as before, 
\begin{align}\nonumber
   \mathfrak{v}(\varphi, E(\bp)^{-1}\chi E(\bp)^{k+1}f) 
  &= (\varphi, V E(\bp)^{-1}\chi E(\bp)^{k+1}f)\\
  &= (\tilde{\chi} \varphi, V E(\bp)^{-1}\chi
  E(\bp)^{k+1}f) \label{b2a}\\ 
  &\quad+ ((1-\tilde{\chi}) \varphi, V
  E(\bp)^{-1}\chi E(\bp)^{k+1}f)\,. 
    \nonumber 
\end{align} 
Since $\tilde{\chi}$ has support away from zero, $V \tilde{\chi}
\varphi \in H^k(\mathbb{R}^3)$ and hence there exists $w_3 \in
L^2(\mathbb{R}^3)$ such that 
\begin{equation*}
  (\tilde{\chi} \varphi, V E(\textbf{p})^{-1} \chi E(\textbf{p})^{k+1}
  f)= (w_3, f)\,. 
\end{equation*}
For the second term in \eqref{b2a} we proceed via an approximation. Let
$\{\varphi_{n}\}_{n=1}^\infty \subset C^{\infty}_0(\mathbb{R}^3)$ such that
$\varphi_{n} \rightarrow \varphi, n\to\infty$, in $L^2(\mathbb{R}^3)$. Since
$(1-\tilde{\chi})V E(\textbf{p})^{-1}\chi E(\textbf{p})^{k+1}f$
belongs to $L^2(\mathbb{R}^3)$, we have that
\begin{align*}
  (\varphi,(1-\tilde{\chi}) V E(\bp)^{-1}\chi
  E(\bp)^{k+1}f)= 
  \lim_{n \rightarrow +\infty} (\varphi_{n},(1-\tilde{\chi}) V
  E(\bp)^{-1}\chi E(\bp)^{k+1}f)\,. 
\end{align*}
For each $n\in \mathbb{N}$, \(V(1-\tilde{\chi})\varphi_{n}\in
H^{m}(\rz^3)\) for all \(m\), since \(\varphi_{n}\in
C_0^\infty(\rz^3)\), and \(V\) maps \(H^{k}(\rz^3)\) into
\(H^{k-1}(\rz^3)\) for all \(k\). Therefore, 
\(E(\bp)^{k+1}\chi E(\bp)^{-1}V(1-\tilde{\chi})\varphi_{n}\in
L^2(\rz^3)\), and so
\begin{align*}
  (\varphi_{n}&,(1-\tilde{\chi}) V E(\textbf{p})^{-1}\chi
  E(\textbf{p})^{k+1}f)\\ 
  &= ( E(\textbf{p})^{k+1} \chi E(\textbf{p})^{-1} V
  (1-\tilde{\chi})\varphi_{n},f)\\ 
  &=( E(\textbf{p})^{k+1} \chi E(\textbf{p})^{-1} (1-\tilde{\chi})
  E(\textbf{p})E(\textbf{p})^{-1}  V \varphi_{n},f)\,. 
\end{align*}
Here $E(\textbf{p})^{-1} V$ is bounded by \eqref{Kato's ineq}, and $\chi
E(\bp)^{-1}(1-\tilde{\chi})$ is a smoothing operator by the
choice of the supports of $\chi$ and $\tilde{\chi}$. It then follows
that $\{E(\textbf{p})^{k+1} \chi E(\textbf{p})^{-1} (1-\tilde{\chi})
E(\textbf{p})E(\textbf{p})^{-1}  V \varphi_{n}\}_{n\in\nz}$ is a uniformly
bounded sequence in $L^2(\mathbb{R}^3)$ and hence there exists $w_4
\in L^2(\mathbb{R}^3)$ such that 
\begin{equation*}
  \lim_{n \rightarrow +\infty} (\varphi_{n},((1-\tilde{\chi}) V
  E(\textbf{p})^{-1}\chi E(\textbf{p})^{k+1}f)=(w_4,f)\,. 
\end{equation*}

For the third term in \eqref{terms}, we have to separate the cases
$k=0$ and $k \geq 1$.

Let $k=0$. The terms $R_{\gamma^{\rm HF}} \varphi$ and $K_{\gamma^{\rm HF}}
\varphi$ belong to $L^2(\mathbb{R}^3)$, since \(R_{\gamma^{\rm
    HF}}\in L^{\infty}(\rz^3)\) (see \eqref{eq:propR-gamma}) and
\(K_{\gamma^{\rm HF}}\) is Hilbert-Schmidt (see Lemma~\ref{spectrum}), 
and therefore
\begin{equation*}
  \mathfrak{b}_{\gamma^{\rm HF}}(\varphi, E(\bp)^{-1}\chi E(\bp)f)=
  \alpha\,(E(\bp) \chi E(\bp)^{-1}(
  R_{\gamma^{\rm HF}}-K_{\gamma^{\rm HF}}) \varphi,f)\,. 
\end{equation*}
Assume now $k\geq 1$. With $\tilde{\chi}$ as before,
\begin{align}\label{S2} 
  \mathfrak{b}_{\gamma^{\rm HF}}(\varphi,E(\bp)^{-1}\chi
  E(\bp)^{k+1}f)
  &=\alpha\,(\tilde{\chi}(R_{\gamma^{\rm HF}}-K_{\gamma^{\rm HF}})\varphi,
  E(\bp)^{-1}\chi E(\bp)^{k+1}f)\\ 
  &\ +\alpha\,((1-\tilde{\chi})(R_{\gamma^{\rm HF}}-K_{\gamma^{\rm HF}}) \varphi,
  E(\bp)^{-1}\chi E(\bp)^{k+1}f)\,.
  \nonumber 
\end{align}
By the induction hypothesis and Lemma~\ref{regfj} (see
Appendix~\ref{app:lemmata})
we have that
$\tilde{\chi}R_{\gamma^{\rm HF}} \varphi$ and
$\tilde{\chi}K_{\gamma^{\rm HF}}\varphi$ belong to $H^k(\mathbb{R}^3)$. 
Therefore there exists $w_5 \in L^2(\mathbb{R}^3)$ such that
\begin{equation*}
  (\tilde{\chi}(R_{\gamma^{\rm HF}}-K_{\gamma^{\rm HF}}) \varphi,
  E(\bp)^{-1}\chi E(\bp)^{k+1}f) = (w_5,f)\,. 
\end{equation*}
For the second term in \eqref{S2} we find, since
$R_{\gamma^{\rm HF}}\varphi , \, K_{\gamma^{\rm HF}}\varphi\in
L^2(\mathbb{R}^3)$, that
\begin{align*}
  ((1-\tilde{\chi})(R_{\gamma^{\rm HF}}&-K_{\gamma^{\rm HF}}) \varphi,
  E(\bp)^{-1}\chi E(\bp)^{k+1}f)\\ 
  &=(\chi
  E(\bp)^{-1}(1-\tilde{\chi})(R_{\gamma^{\rm HF}}-K_{\gamma^{\rm HF}})
  \varphi, E(\bp)^{k+1}f)\,, 
\end{align*}
and the result follows since $\chi E(\mathbf{p})^{-1}
(1-\tilde{\chi})$ is a smoothing operator. 
\end{proof}

\subsection{Exponential decay of the Hartree-Fock orbitals}

The pointwise exponential decay \eqref{eq:decay} will be a consequence
of Proposition~\ref{lem:L2-bound} and Lemma~\ref{invt} below.
\begin{proposition}\label{lem:L2-bound}
 Let \(\gamma^{\rm HF}\) be a Hartree-Fock minimizer, let
 $h_{\gamma^{\rm HF}}$ be the corresponding Hartree-Fock operator as defined in  
 \eqref{hgamma}, and let \(\{\varphi_{i}\}_{i=1}^N\) be  the Hartree-Fock
orbitals, such that
 \begin{align*}
   h_{\gamma^{\rm HF}}\varphi_{i}=\varepsilon _{i}\varphi_{i}\ ,\ i=1,\ldots,N\,, 
 \end{align*}
 with $0>\varepsilon_N\geq\ldots\geq\varepsilon_{1}>{}-\alpha^{-1}$
the \(N\) lowest eigenvalues of \(h_{\gamma^{\rm HF}}\). 
\begin{enumerate}
\item[{\rm (i)}]Let
$\nu_{\varepsilon_N}:=\sqrt{-\varepsilon_N(2\alpha^{-1}+\varepsilon_N)}$.
Then $\varphi_i \in \mathcal{D}(e^{\beta|\,\cdot\,|})$ for every
$\beta<\nu_{\varepsilon_N}$ and \(i\in\{1,\ldots,N\}\).
\item[{\rm (ii)}] Assume \(h_{\gamma^{\rm HF}}\varphi=\varepsilon\varphi\)
  for some \(\varepsilon\in[\varepsilon_{N},0)\), and let
$\nu_{\varepsilon}:=\sqrt{-\varepsilon(2\alpha^{-1}+\varepsilon)}$.
Then $\varphi\in \mathcal{D}(e^{\beta|\,\cdot\,|})$ for every
$\beta<\nu_{\varepsilon}$.
\end{enumerate}
\end{proposition}

\begin{lemma}\label{invt}
Let \(E<0\)
and
$\nu_{E}:=\sqrt{|-E(2\alpha^{-1}+E)|}=\sqrt{|\alpha^{-2}-(E+\alpha^{-1})^2|}$. 

Then the operator $T(-{\rm i}\nabla)-E=\sqrt{-\Delta
  +\alpha^{-2}}-\alpha^{-1}-E$ is invertible and the integral kernel
of its inverse is given by  
\begin{align}\label{eq:kernel}
  (T-E)^{-1}(\bx,\by)&=G_{E}(\bx-\by)=
  \frac{(E+\alpha^{-1})
  e^{-\nu_{E}|\bx-\by|}}{4\pi
  |\bx-\by|} + \frac{\alpha^{-1}}{2
  \pi^2}\frac{K_1(\alpha^{-1}|\bx-\by|)}{|\bx-\by|}
  \nonumber
   \\ 
  &\quad+(\alpha^{-2}-\nu_{E}^2)\frac{\alpha^{-1}}{2 \pi^2}
  \Big[\frac{K_1(\alpha^{-1}|\cdot|)}{|\cdot|} *
  \frac{e^{-\nu_{E}|\,\cdot\,|}}{4\pi|\cdot|}\Big](\bx-\by)\,,
\end{align}
where $K_1$ is a modified Bessel function of the second
kind \cite{AbraSte}.

Moreover, 
\begin{align}\label{est1}
  0\leq G_E(\bx)&\leq
  C_{\alpha,E}\frac{
  e^{-\nu_{E}|\bx|}}{4\pi|\bx|} 
  + \frac{\alpha^{-1}}{2
  \pi^2}\frac{K_1(\alpha^{-1}|\bx|)}{|\bx|}\,, \\\label{est2}
  e^{\beta|\,\cdot\,|}G_E&\in L^q(\rz^3) \ \text{ for all }\ 
  \beta<\nu_{E} \ \text{ and }\ q\in [1,3/2)\,.
\end{align}
\end{lemma}
\begin{pf*}{\it Proof of Lemma~\ref{invt}}
The formula \eqref{eq:kernel} for the kernel of $(T-E)^{-1}$ can be
found in \cite[eq.\ (35)]{SS}.

The estimate \eqref{est1} is a consequence of the bound
\begin{align*}
  \frac{K_1(\alpha^{-1}|\cdot|)}{|\cdot|} *
  \frac{e^{-\nu_{E}|\,\cdot\,|}}{4\pi|\cdot|}(\bx) \leq
  C_{\alpha,E} \frac{
  e^{-\nu_{E}|\bx|}}{4 \pi |\bx|}\,. 
\end{align*}
This estimate, on the other hand, follows from Newton's theorem (see
e.\ g.\ \cite{LiebLoss}),
\begin{align*}
   &\int_{\rz^3}
   \frac{K_1(\alpha^{-1}|\bx-\by|)}{|\bx-\by|}
   \,\frac{e^{-\nu_{E}|\by|}}{4\pi|\by|}
   \,d\by\\ 
   & \leq  e^{-\nu_{E}|\bx|}\int_{\rz^3}
   \frac{K_1(\alpha^{-1}|\bx-\by|)}{|\bx-\by|}
   \,\frac{e^{\nu_{E}|\bx-\by|}}{4\pi|\by|}
   \,d\by
   \leq  \frac{e^{-\nu_{E}|\bx|}}{4\pi|\bx|}
  \int_{\mathbb{R}^3}
  \frac{K_1(\alpha^{-1}|\bz|)}{|\bz|}
  \,e^{\nu_{E}|\bz|}
  \,d \mathbf{z}\,. 
\end{align*}
The last integral is finite since $\nu_{E}<\alpha^{-1}$, 
using the following properties of $K_1$ (see
\cite[8.446, 8.451.6]{GradshteynRyzhik1980}): 
\begin{equation}\label{eq:app2}
  K_1(t) \leq \frac{1}{|t|} \ \text{ for all } \ t>0\,,
\end{equation}
and for every $r>0$ there exists $c_r$ such that
\begin{equation}\label{eq:app3}
  K_1(t) \leq c_r \frac{e^{-t}}{\sqrt{t}} \ \text{ for all } \ t\ge r\,.
\end{equation}
The estimate \eqref{est2} is a consequence of \eqref{est1},
\eqref{eq:app2}, and \eqref{eq:app3}.
\end{pf*}

Before proving Proposition~\ref{lem:L2-bound}, we apply it, and
Lemma~\ref{invt}, to prove the pointwise exponential decay, i.e., the
estimate in \eqref{eq:decay}. 

\begin{pf*}{Proof of Theorem~\ref{HF2} {\rm (iii)}} 
Fix $i \in \{1,\dots, N\}$. If $Z\alpha<1/2$ we can rewrite the
Hartree-Fock equation \eqref{eq:HF-equations} as
\begin{align}\label{point}
  \big(\sqrt{-\Delta +\alpha^{-2}}-\alpha^{-1}\big) \varphi_i=
  \varepsilon_i\varphi_i +\frac{Z \alpha}{|\bx|}\,
  \varphi_i - \alpha R_{\gamma^{\rm HF}} \varphi_i + \alpha
  K_{\gamma^{\rm HF}} \varphi_i\,.  
\end{align}
The idea of the proof is to study the elliptic regularity of the
corresponding parametrix. By Lemma \ref{invt} we find that
\begin{align*}
  \varphi_i(\bx) =\int_{\rz^3} 
  (T-\varepsilon_N)^{-1}(\bx,\by)\big[(\varepsilon_i
  -\varepsilon_N)\varphi_i+\frac{Z\alpha}{|\cdot|}\,\varphi_i -
  \alpha R_{\gamma^{\rm HF}}\varphi_i+\alpha K_{\gamma^{\rm HF}}
  \varphi_i\big](\by)
  \,d\by\,.
\end{align*}

In the case $1/2\leq Z \alpha<2/\pi$, on the other hand, the operator
of which we are studying the eigenfunctions cannot be written as a sum
of operators acting on $L^2(\rz^3)$ and hence we cannot write directly
the equation \eqref{eq:HF-equations} as in \eqref{point}. However,
since the eigenfunctions are smooth away from the origin we are able
to write a pointwise equation for a localized version of
$\varphi_i$. In fact,
let $\chi \in
C^{\infty}(\mathbb{R}^3)$ be such that $0 \leq \chi
\leq 1$ and  
\begin{equation*}
  \chi(\mathbf{x})=\left\{ \begin{array}{ll}
  1 & \mbox{ if } |\mathbf{x}|\geq 1,\\ 
  0 & \mbox{ if } |\mathbf{x}| \leq 1/2\,,
\end{array} \right.
\end{equation*}
and let, for \(R>0\), \(\chi_R(\bx)=\chi(\bx/R)\). 
We will derive an equation (similar to \eqref{point}) for
$T(-{\rm i}\nabla)(\chi_R \varphi_i)$. Indeed, for every $u \in 
H^{1/2}(\mathbb{R}^3)$ we have that
\begin{align*}
  (u,h_{\gamma^{\rm HF}}(\chi_R \varphi_i))&=\mathfrak{e}(u,\chi_R
  \varphi_i)-\alpha^{-1}(u,\chi_R\varphi_i)-\mathfrak{v}(u,\chi_R \varphi_i)+
  \mathfrak{b}_{\gamma^{\rm HF}}(u,\chi_R \varphi_i)\\ 
  &=(\chi_R u,h_{\gamma^{\rm HF}}\varphi_i)+\mathfrak{e}(u,\chi_R
  \varphi_i)-\mathfrak{e}(\chi_R u,\varphi_i) \\ 
  &\ \ +\mathfrak{b}_{\gamma^{\rm HF}}(u, \chi_R \varphi_i) 
  -\mathfrak{b}_{\gamma^{\rm HF}}(\chi_R u,\varphi_i)\,. 
\end{align*}
Note that
\begin{align*}
  \mathfrak{e}(u,\chi_R \varphi_i)-\mathfrak{e}(\chi_R u, \varphi_i)=(u,
  [E(\mathbf{p}),\chi_R] \varphi_i)\,, 
\end{align*}
where $[E(\mathbf{p}),\chi_R]$ is a bounded operator in
$L^2(\mathbb{R}^3)$ (see Appendix~\ref{app:pseudo}), and
\begin{equation*}
  \mathfrak{b}_{\gamma^{\rm HF}}(u,\chi_R\varphi_i)
  -\mathfrak{b}_{\gamma^{\rm HF}} 
  (\chi_R u,\varphi_i)=(u,\mathcal{K}\varphi_i)\,, 
\end{equation*}
with $\mathcal{K}$ the bounded operator on $L^2(\mathbb{R}^3)$ given
by the kernel
\begin{align}\label{eq:kernelSkriptK}
  \mathcal{K}(\bx,\by)=\alpha\sum_{j=1}^N \varphi_j
  (\bx)\overline{\varphi_j(\by)}\,\frac{\chi_R
  (\bx)-\chi_R(\by)}{|\bx-\by|}\,.  
\end{align} 
Therefore there exists $w \in L^2(\mathbb{R}^3)$ such that
\begin{align*}
  \mathfrak{e}(u,\chi_R \varphi_i)&=(\varepsilon_i+\alpha^{-1}) 
  (u,\chi_R\varphi_i)+\mathfrak{v}(u,\chi_R\varphi_i)
  -\mathfrak{b}_{\gamma^{\rm HF}}(u,\chi_R\varphi_i)\\ 
  &\ \ +(u,[E(\mathbf{p}),\chi_R]\varphi_i)+(u,\mathcal{K}\varphi_i)
  =(u,w)\,.
\end{align*}
Hence $\chi_R \varphi_i \in H^1(\mathbb{R}^3)$ and we can write the
pointwise equation  
\begin{align}\label{eq:locPoint}\nonumber
  (\sqrt{-\Delta+\alpha^{-2}}-\alpha^{-1})\chi_R\varphi_i
  &=\varepsilon_i\chi_R\varphi_i 
  +\frac{Z\alpha}{|\bx|}\chi_R\varphi_i-\alpha R_{\gamma^{\rm HF}} 
  \chi_R\varphi_i \\ 
  &\ \ +\alpha K_{\gamma^{\rm HF}}(\chi_R\varphi_i)+[E(\bp),\chi_R]
  \varphi_i+\mathcal{K}\varphi_i\,.
\end{align}
This is the substitute for \eqref{point} in the case $1/2\leq Z
\alpha<2/\pi$; if \(Z\alpha<1/2\), the proof below simplifies
somewhat, using \eqref{point} directly.

By Lemma~\ref{invt}, \eqref{eq:locPoint} implies that 
\begin{align}\label{eq:inverted}\nonumber
  \chi_R(\bx)\varphi_i(\bx)&=\int_{\rz^3}
  (T-\varepsilon_N)^{-1}(\bx,\by)\big[\frac{Z
  \alpha}{|\cdot|}\chi_R\varphi_i-\alpha R_{\gamma^{\rm HF}}\chi_R
  \varphi_i+\alpha K_{\gamma^{\rm HF}}(\chi_R \varphi_i) \\ 
  &\qquad\qquad\qquad
  +(\varepsilon_i- \varepsilon_N)\chi_R\varphi_i+
  [E(\bp),\chi_R]\varphi_i+\mathcal{K}\varphi_i \big]
  (\by)\,d\by\,. 
\end{align}

We will first show that, for all \(R>0\) and
\(\beta<\nu_{\varepsilon_N}\),  
\begin{align}\label{eq:firstBoot}
  \chi_R\varphi_ie^{\beta|\,\cdot\,|}\in
  L^p(\rz^3)+L^{\infty}(\rz^3) \ \text{ for }p\in[2,6)\,,
\end{align}
and then, by a
bootstrap argument, that \(\chi_R\varphi_ie^{\beta|\,\cdot\,|}\in
L^{\infty}(\rz^3)\), which is the claim of
Theorem~\ref{HF2} {\rm (iii)}.

We multiply \eqref{eq:inverted} by
\(\chi_{R/2}(\bx)e^{\beta|\bx|}\). 
Using that \(|(Z\alpha/|\by|)\chi_R(\by)|\le (Z\alpha)/R\) for all \(\by\in\rz^3\),
\eqref{eq:propR-gamma}, \eqref{def:K-gamma}, and \eqref{eq:kernelSkriptK}
(recall \eqref{eq:formMinimizer}, that \(\varphi_j\in H^{1/2}(\rz^3)\), and
\eqref{eq:Kato}) we get, for some constant $C=C_{R,\alpha}>0$, that  
\begin{align}\notag
  \big|\chi_R(\bx)\varphi_i(\bx)&e^{\beta|\bx|}\big|\leq
  C \chi_{R/2}(\bx)e^{\beta|\bx|}\int_{\rz^3}(T-\varepsilon_N)^{-1}
  (\bx,\by)\big[|\varphi_i(\by)|
  +\sum_{j=1}^N|\varphi_j(\by)|\big]\,d\by \\ 
  &+ \chi_{R/2}(\bx)e^{\beta|\bx|}
  \Big|\int_{\rz^3}
  (T-\varepsilon_N)^{-1}(\bx,\by)\big([E(\bp),\chi_R]
  \varphi_i\big)(\by)\,d\by \Big|\,. 
  \label{smile} 
\end{align}
We will show that the first term on the right side of \eqref{smile}
belongs to \(L^{p}(\mathbb{R}^3)\) for
$p\in[2,6)$, and that the second belongs to
\(L^{\infty}(\mathbb{R}^3)\). This will prove \eqref{eq:firstBoot}. 

The first term
on the right side of \eqref{smile} is a sum of terms of the form 
\begin{equation}\label{hf}
  h_{f}(\mathbf{x}):=\chi_{R/2}(\bx) e^{\beta|\bx|}
  \int_{\rz^3}
  (T-\varepsilon_{N})^{-1}(\bx,\by)\,|f(\by)|  
  \,d\by\,,
\end{equation} 
with $f$ such that, by Proposition~\ref{lem:L2-bound}, $f e^{\beta
  |\,\cdot\,|} \in L^2 (\mathbb{R}^3)$. By Lemma
\ref{invt} we have,  using $e^{|\mathbf{x}|-|\by|} \leq
e^{|\mathbf{x}-\by|}$, that
\begin{align*}
  |h_f(\bx)| & \leq C\int_{\rz^3} 
  e^{\beta|\bx-\by|} G_{\varepsilon_N}(\bx-\by)
  e^{\beta |\by|}|f(\by)|\,d\by\,.
\end{align*} 
From Young's inequality it follows
that $h_{f} \in L^p(\mathbb{R}^3)$ for all $p \in [2,6)$, since
\(\beta<\nu_{\varepsilon_N}\), so (by
Proposition~\ref{lem:L2-bound}) $f e^{\beta
  |\,\cdot\,|}\in L^2(\mathbb{R}^3)$ and (by Lemma
\ref{invt})
$e^{\beta|\,\cdot\,|}G_{\varepsilon_N}\in L^q(\mathbb{R}^3)$ for
all $q \in [1,3/2)$.  

We now prove that the second term on the right side of \eqref{smile} is in
$L^{\infty}(\mathbb{R}^3)$. 
This  
follows from Young's inequality once we have proved that
\begin{align}\label{eq:claim2}
  e^{\beta|\,\cdot\,|}[E(\mathbf{p}),\chi_R] \varphi_i\in
  L^p(\mathbb{R}^3)
  \ \text{ for } p\in[2,\infty)\,,
\end{align}
 since 
\begin{align*} 
   e^{\beta|\bx|}\Big|\int_{\rz^3}&(T-\varepsilon_N)^{-1}
  (\bx,\by)\big([E(\bp),\chi_R]\varphi_i\big)
  (\by)\,d\by\Big|\\  
  \leq \int_{\rz^3}&e^{\beta|\bx-\by|}G_{\varepsilon_N}(\bx-\by)
  e^{\beta|\by|}\big|[E(\bp),\chi_R]\varphi_i\big|
  (\by)\,d\by\,,
\end{align*}
and $e^{\beta|\,\cdot\,|}G_{\varepsilon_N}\in
L^q(\mathbb{R}^3)$ for $q\in[1,3/2)$. 

To prove \eqref{eq:firstBoot} it therefore remains to prove
\eqref{eq:claim2}. 
To do so, we consider
a new localization function. Let $\eta \in C^{\infty}_0(\mathbb{R}^3)$
be 
such that $0 \leq \eta \leq 1$ and 
\begin{align*}
  \eta(\bx)= 
  \begin{cases}
  \ 1 & \text{ if } R/4\leq |\bx|\leq 3R/2 \\
  \ 0 & \text{ if } |\bx|\leq R/8 \text{ or }|\bx|\geq 2R\,, 
  \end{cases}
\end{align*}
and consider the following splitting
\begin{align}\notag 
  & e^{\beta|\,\cdot\,|} [E(\mathbf{p}),\chi_R] \varphi_i 
  = e^{\beta|\,\cdot\,|} \eta
  [E(\mathbf{p}),\chi_R] (\eta \varphi_i)+ e^{\beta|\,\cdot\,|}\eta
  [E(\mathbf{p}),\chi_R] ((1-\eta)\varphi_i) \\ \label{4terms} 
  &\hspace{1cm} +e^{\beta|\,\cdot\,|} (1-\eta)[E(\mathbf{p}),\chi_R]
  (\eta\varphi_i)+e^{\beta|\,\cdot\,|} (1-\eta) [E(\mathbf{p}),\chi_R] (1-\eta)
  \varphi_i\,.  
\end{align}
Since $\eta \varphi_i \in H^k(\mathbb{R}^3)$ for all
\(k\in\mathbb{N}\) (as proved earlier), $[E(\mathbf{p}),\chi_R]
(\eta \varphi_i)$ belongs to $H^k(\mathbb{R}^3)$ for all
\(k\in\mathbb{N}\). Hence, since
$\eta$ has compact support away from \(\bx=0\), the first term on the
right side of 
\eqref{4terms} is in $L^p(\mathbb{R}^3)$ for $p\in[1,\infty]$ 
by Sobolev's imbedding theorem (the term is smooth).

For the second term in \eqref{4terms} we proceed by duality: We will
prove that
\begin{align*}
  \psi(\bx):=\big(e^{\beta |\,\cdot\,|}\eta
  [E(\bp),\chi_R]((1-\eta)\varphi_i)\big)(\bx)
\end{align*}
defines a bounded linear functional on \(L^{q}(\rz^3)\) for any
\(q\in(1,2]\). It then follows that \(\psi\in L^p(\rz^3)\) for all
\(p\in[2,\infty)\).  

Note that \cite[7.12 Theorem (iv)]{LiebLoss}
\begin{align}\label{eq:formula T}\nonumber
  (g,&[\sqrt{-\Delta+\alpha^{-2}}-\alpha^{-1}]g) \\
  &=\frac{\alpha^{-2}}{4\pi^2}\int_{\rz^3}\int_{\rz^3}\frac{|g(\bx)-g(\by)|^2}{|\bx-\by|^2}
  K_2(\alpha^{-1}|\bx-\by|)\,d\bx d\by \ \text{ for }
  g\in\mathcal{S}(\rz^3)\,, 
\end{align}
where \(K_2\) is a modified Bessel function of the second kind (in
fact, \(K_2(t)=-t \frac{d}{dt}[t^{-1}K_1(t)]\)), satisfying \cite{AbraSte}
\begin{align}\label{eq:est-K2}
  K_2(t)\le Ct^{-1}e^{-t}\ \text{ for }\ t\ge 1\,.
\end{align}
Let $f \in C_0^\infty(\rz^3)$. Using \eqref{eq:formula T} and
polarization, we have that 
\begin{align*}
  &\int_{\rz^3}\overline{f(\bx)}\psi(\bx)\,d\bx=
  (f,e^{\beta |\,\cdot\,|}\eta
  [E(\bp),\chi_R]((1-\eta)\varphi_i))
  \\&=\frac{\alpha^{-2}}{4\pi^2}\iint_{|\bx-\by|\geq R/4}
  \!\!\!\!
  \frac{\chi_R(\bx)-\chi_R(\by)}{|\bx-\by|^2}
  \,K_2(\alpha^{-1}|\bx-\by|)
  \\&\qquad\quad\times\big[\,\overline{f(\bx)}e^{\beta|\bx|}
  \eta(\mathbf{x})(1-\eta(\by))\varphi_i(\by)
  -\overline{f(\by)}e^{\beta |\by|}
  \eta(\by)(1-\eta(\bx))\varphi_i(\bx)\big]\,d\bx d\by\,, 
\end{align*}
by the properties of \(\chi\) and $\eta$. Hence,
\begin{align}\label{one-bound}\nonumber
  \Big|\int_{\rz^3}&\overline{f(\bx)}\psi(\bx)\,d\bx\Big|
  \\
  & \leq  C_R \iint_{|\bx-\by|\geq R/4}
  |f(\bx)| e^{\beta |\bx-\by|}
  K_2(\alpha^{-1}|\bx-\by|)  e^{\beta
  |\by|}|\varphi_i(\by)| \,d\bx d\by\,,
 \nonumber \\
 & \leq  C_R \iint
  |f(\bx)| e^{\beta |\bx-\by|}
  K_2(\alpha^{-1}|\bx-\by|)\chi_{R/4}(|\bx-\by|)  e^{\beta
  |\by|}|\varphi_i(\by)| \,d\bx d\by\,.
\end{align}
Note that, since \(\beta<\nu_{\varepsilon_N}<\alpha^{-1}\), \eqref{eq:est-K2}
implies that
$e^{\beta|\,\cdot\,|} K_2(\alpha^{-1}|\cdot|)\chi_{R/4}$ is in
$L^r(\rz^3)$ for all $r \geq 1$. Since (by
Proposition~\ref{lem:L2-bound}) \(e^{\beta|\,\cdot\,|}\varphi_i\in
L^2(\rz^3)\), Young's inequality therefore gives that
\begin{align*}
  (e^{\beta|\,\cdot\,|}K_2(\alpha^{-1}|\cdot|)
  \chi_{R/4})*(e^{\beta|\,\cdot\,|}|\varphi_i|)
  \in L^{s}(\rz^3)\ \text{ for all } s\in[2,\infty)\,.
\end{align*}
This, \eqref{one-bound}, and H\"older's inequality (with $1/q+1/s=1$)
imply that, for all \(f\in C_0^\infty(\rz^3)\) and all \(q\in(1,2]\)
\begin{align*} 
   \Big|\int_{\rz^3}&\overline{f(\bx)}\psi(\bx)\,d\bx\Big|
  \leq C_R \big\|(e^{\beta|\,\cdot\,|}K_2(\alpha^{-1}|\cdot|)
  \chi_{R/4})*(e^{\beta|\,\cdot\,|}|\varphi_i|)\big\|_{s} \,\|f\|_{q}\,.
\end{align*}
By density of \(C_0^\infty(\rz^3)\) in \(L^{q}(\rz^3)\), it follows
that \(\psi\) defines a  bounded linear 
functional on \(L^{q}(\rz^3)\) for any 
\(q\in(1,2]\), and therefore, that \(\psi\in L^p(\rz^3)\) for all
\(p\in[2,\infty)\).  

Proceeding similarly one shows that the two remaining terms in
\eqref{4terms} are also in $L^p(\mathbb{R}^3)$ for  all
\(p\in[2,\infty)\). 

This finishes the proof of \eqref{eq:claim2}, and therefore of
\eqref{eq:firstBoot}.

Finally we prove that $\chi_R \varphi_i e^{\beta|\,\cdot\,|} \in
L^{\infty}(\rz^3)$. We start again from \eqref{smile}.
We already know that the second
term is in $L^{\infty}(\rz^3)$. The first term is a sum of terms
of the form (see also \eqref{hf})
\begin{equation*}
  h_f(\bx)= \chi_{R/2}(\bx) e^{\beta|\bx|}
  \int_{\rz^3} (T-\varepsilon_{N})^{-1} (\bx,\by)
  |f(\by)|\,d\by\,,  
  \end{equation*} 
with \(f\in L^2(\rz^3)\) and $\chi_{R/4} e^{\beta|\,\cdot\,|}f \in
L^p(\mathbb{R}^3)+L^{\infty}(\mathbb{R}^3)$ for $p\in[2,6)$ by what just
proved, replacing $R$ by $R/4$ in \eqref{eq:firstBoot}. We find that 
\begin{align*}
  h_f(\bx)&\leq\chi_{R/2}(\bx) \int_{\rz^3}
  e^{\beta|\bx-\by|} (T-\varepsilon_{N})^{-1}
  (\bx,\by)  e^{\beta|\by|} \chi_{R/4}(\by)
  |f(\by)|\,d\by \\ 
  & + \chi_{R/2}(\bx) \int_{\rz^3}
  e^{\beta|\bx-\by|} (T-\varepsilon_{N})^{-1}
  (\bx,\by)  e^{\beta|\by|} (1-\chi_{R/4})(\by)
  |f(\by)|\,d\by\,, 
\end{align*}
and, again by Young's inequality, we see that both terms are in
$L^{\infty}(\mathbb{R}^3)$. Notice that in the second integrand
$|\bx-\by|>R/4$.

This finishes the proof of Theorem~\ref{HF2} {\rm (iii)}.
\end{pf*}
It therefore remains to prove Proposition~\ref{lem:L2-bound}.
\begin{pf*}{Proof of Proposition~\ref{lem:L2-bound}}
We start by proving (i).
It will be convenient to write the Hartree-Fock equations
\(h_{\gamma^{\rm HF}}\varphi_i=\varepsilon_i\varphi_i\),
  \(i=1,\ldots,N\), (see \eqref{eq:HF-equations}) as a system. 

Let $\mathtt{t}$ be the quadratic form with domain
$[H^{1/2}(\mathbb{R})]^N$ defined by
\begin{align*}
  \mathtt{t}(\mathrm{u},\mathrm{v})=\sum_{i=1}^N
  \mathfrak{t}(\mathrm{u}_i,\mathrm{v}_i) \mbox{ for all }\mathrm{u},
  \mathrm{v}\in 
  [H^{1/2}(\mathbb{R}^3)]^N\,, 
\end{align*}
where $\mathrm{u}_i$ denotes the $i$-th component of $\mathrm{u} \in
[H^{1/2}(\mathbb{R}^3)]^N$ and $\mathfrak{t}$ is the quadratic
form defined in \eqref{a}. Similarly we define the quadratic forms
$\mathtt{v}$, $\mathtt{r}_{\gamma}$ and $\mathtt{k}_{\gamma}$, all with
domain $[H^{1/2}(\mathbb{R}^3)]^N$, by
\begin{align*}
  \mathtt{v}(\mathrm{u},\mathrm{v})&=\sum_{i=1}^N
  \mathfrak{v}(\mathrm{u}_i,\mathrm{v}_i)\,, \ \  
  \mathtt{r}_{\gamma}(\mathrm{u},\mathrm{v})=\alpha\sum_{i=1}^N
  (\mathrm{u}_i,R_{\gamma}\mathrm{v}_i)\,, \ \ 
  \mathtt{k}_{\gamma}(\mathrm{u},\mathrm{v})=\alpha\langle
  \mathrm{u},\mathrm{K}_{\gamma} \mathrm{v} \rangle\,, 
\end{align*}
with $\mathfrak{v}$ defined in \eqref{b}, $R_{\gamma}$ defined in \eqref{Rgamma},
and $\mathrm{K}_{\gamma}$ the $N\times N$-matrix given by
\begin{align*}
  (\mathrm{K}_{\gamma})_{i,j}=\int_{\mathbb{R}^3}
  \frac{\varphi_i(\by)\overline{\varphi_j(\by)}}{|\bx-\by|} 
  \,d\by\,. 
\end{align*}
The effect of writing the Hartree-Fock equations as a system is that
\(\mathrm{K}_\gamma\) is a (non-diagonal) multiplication
operator. This idea was already used in \cite{LiebSimonHF}. Note that 
\((\mathrm{K}_{\gamma})_{i,j}\in L^3(\rz^3)\cap L^\infty(\rz^3)\); the
argument is the same as for \eqref{Rgamma}. 
 
Let finally $\mathtt{E}$ be the $N \times N$ matrix defined by
$(\mathtt{E})_{i,j}=-\varepsilon_i\delta_{i,j}$.  

We then define the quadratic form $\mathtt{q}$ by
\begin{align}\label{mathttq}
  \mathtt{q}(\mathrm{u},\mathrm{v})=
  \mathtt{t}(\mathrm{u},\mathrm{v})
  -\mathtt{v}(\mathrm{u},\mathrm{v})
  +\mathtt{r}_{\gamma}(\mathrm{u},\mathrm{v})
  -\mathtt{k}_{\gamma}(\mathrm{u},\mathrm{v})+\langle
  \mathrm{u},\mathtt{E} \mathrm{v} \rangle\,. 
\end{align}
One sees that the quadratic form domain of $\mathtt{q}$ is
$[H^{1/2}(\mathbb{R}^3)]^N$, that $\mathtt{q}$ is closed
(since $\mathtt{t}$ is closed), and that there exists a unique
selfadjoint operator $\mathtt{H}$ with 
$\mathcal{D}(\mathtt{H})\subset[H^{1/2}(\mathbb{R}^3)]^N$ such that 
\begin{equation*}
  \langle \mathrm{u},\mathtt{H}\mathrm{v} \rangle
  =\mathtt{q}(\mathrm{u},\mathrm{v})\ \text{ for all }\ \mathrm{u}\in
[H^{1/2}(\mathbb{R}^3)]^N\,,\, \mathrm{v}\in\mathcal{D}(\mathtt{H})\,.
\end{equation*} 
Notice that the vector
$\Phi=(\varphi_1, \dots, \varphi_N)$ satisfies $\mathtt{H} \Phi=0$. 

Let $W(\kappa)$, $\kappa \in \mathbb{C}^3$, denote the multiplication
operator from a subset of $[L^2(\mathbb{R}^3)]^N$ to
$[L^2(\mathbb{R}^3)]^N$ given by $f(\bx) \mapsto  e^{i \kappa
  \cdot \bx}f(\bx)$. Instead of proving directly the
claim of the proposition, we are going to prove the following
statement, which implies the proposition: 
\begin{equation}\label{eq:eqL2-decay}
  \Phi \in \mathcal{D}(W(\kappa))\ \text{ for }\ \|\mbox{Im}(\kappa)\|_{\rz^3} <
  \nu_{\varepsilon_N}\,, 
\end{equation}
where $\Phi=(\varphi_1, \dots, \varphi_N)$. Here, $\kappa= \mbox{Re}
(\kappa)+i \mbox{Im}(\kappa)$ with $\mbox{Re}(\kappa), \mbox{Im}(\kappa) \in
\mathbb{R}^3$.  

We know that $W(\kappa) \Phi$ is well defined on
$[L^2(\mathbb{R}^3)]^N$ for $\kappa \in \mathbb{R}^3$ and we need to
show that it has a continuation into the `strip'
\(\Sigma_{\nu_{\varepsilon_N}}\), where
$$
  \Sigma_{t}:=\{\kappa \in \mathbb{C}^3\,|\,
  \|\mbox{Im}(\kappa)\|_{\rz^3}< t \}\,.
$$
We shall also need \(\Sigma_{\alpha^{-1}}\); note that
  \(\Sigma_{\alpha^{-1}} \supset \Sigma_{\nu_{\varepsilon_N}}\).   
The idea is to use O'Connor's Lemma (see Lemma \ref{lem:O'Connor} below).

Starting from the quadratic form $\mathtt{q}$ defined
in (\ref{mathttq}) we define the following family of quadratic forms
on $[H^{1/2}(\mathbb{R}^3)]^N$: 
\begin{equation*}
  \mathtt{q}(\kappa)(\mathrm{u},\mathrm{u})
  :=\mathtt{q}(W(-\kappa) \mathrm{u}, W(-\kappa) \mathrm{u})\,,
\end{equation*}
depending on the {\it real} parameter $\kappa \in\rz^3$.
From the
definition, 
\begin{equation*}
  \mathtt{q}(\kappa)(\mathrm{u},\mathrm{u})=\mathtt{t}(\kappa)
  (\mathrm{u},\mathrm{u})-\mathtt{v}(\mathrm{u},\mathrm{u})
  +\mathtt{r}_{\gamma}(\mathrm{u},\mathrm{u})
  -\mathtt{k}_{\gamma}(\mathrm{u},\mathrm{u})+\langle
  \mathrm{u},\mathtt{E} \mathrm{u}\rangle\,, 
\end{equation*}
where
\begin{align}\label{eq:sectQ} 
  \mathtt{t}(\kappa)(\mathrm{u},\mathrm{u})=\sum_{i=1}^N \int_{\rz^3}
  \big(
  \alpha^{-2}+\sum_{j=1}^3 (p_j -\kappa_j )^2\big)^{1/2}
  |{\hat{\mathrm{u}}}_i(\mathbf{p})|^2\,d\bp
  -\alpha^{-1}\langle\mathrm{u},\mathrm{u}\rangle\,.
\end{align}
One sees that $\mathtt{q}(\kappa)$ extends to a family of sectorial forms with angle
$\theta < \frac{\pi}{4}$, and that $\mathtt{q}(\kappa)$ is holomorphic in the
strip $\Sigma_{\alpha^{-1}}$ (indeed, $\|\mbox{Im}(\kappa)\|_{\rz^3}<\alpha^{-1}$ is
needed to assure that the complex number under the square root in
\eqref{eq:sectQ} has
non-negative real part for all \(\bp\in\rz^3\)). Moreover,
$\mathtt{q}(\kappa)$ is 
closed. Indeed, it is sufficient to prove that the real part of
$\mathtt{q}(\kappa)$ is closed, which will follow from
\begin{eqnarray}\label{1}
  \mathtt{v}(\mathrm{u},\mathrm{u})+\mathtt{r}_{\gamma}
  (\mathrm{u},\mathrm{u})+\mathtt{k}_{\gamma}(\mathrm{u},\mathrm{u})+
  \langle \mathrm{u},\mathtt{E} \mathrm{u} \rangle \leq b\,
  \mbox{Re}(\mathtt{t}(\kappa))(\mathrm{u},\mathrm{u}) + K \langle\mathrm{u},
  \mathrm{u} \rangle\,, 
\end{eqnarray}
with $b<1$, $K>0$ and $\mbox{Re}(\mathtt{t}(\kappa))$ closed. We now
prove \eqref{1}. We already know that 
\begin{equation}\label{2}
  \mathtt{r}_{\gamma}(\mathrm{u},\mathrm{u})
  +\mathtt{k}_{\gamma}(\mathrm{u},\mathrm{u})+ \langle
  \mathrm{u},\mathtt{E} \mathrm{u} \rangle  \leq K' \langle \mathrm{u},
  \mathrm{u} \rangle \mbox{ for }K'>0\,. 
\end{equation} 
By \eqref{Kato's ineq} we find
\begin{align}\notag
  \mathtt{v}(\mathrm{u},\mathrm{u})&\leq (Z\alpha)\frac{\pi}{2}
  \sum_{i=1}^N \int_{\rz^3} |\mathbf{p}|\,|{\hat{\mathrm
      u}}_i(\mathbf{p})|^2\,d\bp
  \\ \label{3} 
  & \leq (Z\alpha) \frac{\pi}{2} R \sum_{i=1}^N \big[
  \int_{|\mathbf{p}|\leq R} |{\hat{\mathrm u}}_i(\mathbf{p})|^2\,d \mathbf{p} +
  \int_{|\mathbf{p}|\geq R} |\mathbf{p}|\,|{\hat{\mathrm
      u}}_i(\mathbf{p})|^2\,d\bp\big]\,.   
\end{align}
Let $\delta>0$ be such that $Z \alpha\frac{\pi}{2} (1-\delta)^{-1}<1$. Since
\begin{align*}
  \mbox{Re} (\mathtt{t}(\kappa)) (\mathrm{u},\mathrm{u})&=
  \sum_{i=1}^N \int_{\rz^3}
  \big|\alpha^{-2}+\sum_{j=1}^3 (p_j - \kappa_j)^2\big|^{1/2}
  \cos(\theta(\mathbf{p},\kappa))\,
  |{\hat{\mathrm{u}}}_i(\mathbf{p})|^2  
  \,d\bp
  \\&{}\
  -\alpha^{-1}\langle\mathrm{u},\mathrm{u}\rangle\,, 
\end{align*}
with 
\begin{equation*}
  2 \cos^2(\theta(\mathbf{p},\kappa))-1=
  \frac{\alpha^{-2}+\sum_{j=1}^3(p_j-\mbox{Re}(\kappa_j))^2-(\mbox{Im}(\kappa_j))^2)
  }{|\alpha^{-2}+\sum_{j=1}^3(p_j-\kappa_j)^2|}\,, 
\end{equation*}
there exists $R>0$ such that $\cos(\theta(\mathbf{p},\kappa)) \geq
(1-\delta)$ for $|\mathbf{p}|>R$. Hence we find that
\begin{align} \notag
  \mbox{Re}(\mathtt{t}(\kappa))(\mathrm{u},\mathrm{u})& 
  \geq  (1-\delta)\sum_{i=1}^N 
   \int_{|\bp|>R} \big|\alpha^{-2}+\sum_{j=1}^3 
  (p_j-\kappa_j)^2\big|^{1/2}\, |{\hat{\mathrm u}}_i(\bp)|^2 
  \,d\bp
  \\&{}\quad -\alpha^{-1}\langle\mathrm{u},\mathrm{u}\rangle\ \nonumber
\\ \label{4} 
  &\geq(1-\delta)\sum_{i=1}^N \int_{|\bp|>R}
  (|\mathbf{p}|-C)|\,{\hat{\mathrm u}}_i(\bp) |^2 \,d\bp
 -\alpha^{-1}\langle\mathrm{u},\mathrm{u}\rangle\,, 
\end{align}
with $C>\|\mbox{Re}(\kappa)\|_{\rz^3} $. The estimate in \eqref{1} follows
combining \eqref{2} with \eqref{3} and \eqref{4}. 

The fact that $\mbox{Re}(\mathtt{t}(\kappa))$ is closed follows from
\begin{equation*}
  \frac{1}{\sqrt{2}} \sum_{i=1}^N \int (|\mathbf{p}|-C)\,
  |{\hat{\mathrm u}}_i(\mathbf{p}) |^2\,d\bp \leq
  \mbox{Re}(\mathtt{t}(\kappa))(\mathrm{u},\mathrm{u}) \leq
  \sum_{i=1}^N \int(|\mathbf{p}|+C) 
  |\hat{\mathrm{u}}_i(\mathbf{p})|^2\,d\bp\,, 
\end{equation*} 
with $C\geq2\alpha^{-1}+\mbox{Re}(\kappa)$.

Hence, $\mathtt{q}(\kappa)$ is an analytic family of forms of type (a)
(\cite[p. 395]{Kato}). The associated family $\mathtt{H}(\kappa)$ of 
sectorial operators is a holomorphic family of operators of type (B)
and has domain in a subset of
$[H^{1/2}(\mathbb{R}^3)]^N$. 

We are interested now in locating the essential spectrum of
$\mathtt{H}(\kappa)$. Since \(K_{\gamma}\) is a Hilbert-Schmidt
operator, the essential spectrum of \(\mathtt{H}(\kappa)\) coincides
with the essential spectrum of the operator associated to
\begin{align*}
   {\mathtt t}(\kappa)({\mathrm u},{\mathrm u})-{\mathtt v}({\mathrm u},
  {\mathrm u})+\alpha\,{\mathtt r}_{\gamma}({\mathrm u},{\mathrm u})+
  \langle {\mathrm u},\mathtt{E}{\mathrm u}\rangle\,.
\end{align*}
Notice that the operator associated to this quadratic form is
diagonal. Proceeding as in the proof of \(\sigma_{\rm
  ess}(h_{\gamma})=[0,\infty)\) (Lemma~\ref{spectrum}), one sees that
\(\sigma_{\rm ess}(\mathtt{H}(\kappa))\subset \sigma_{\rm
  ess}(T(\kappa)-\varepsilon_{N})\) with 
\(T(\kappa):=\sqrt{\alpha^{-2}+\sum_{j=1}^{3}(p_j-\kappa_j)^2}-\alpha^{-1}\).
Hence we find that
\begin{equation*}
  \sigma_{\rm ess}(\mathtt{H}(\kappa)) \subset \big\{z \in \mathbb{C}\,\big|\,
  \mbox{Re}(z) \geq \sqrt{\alpha^{-2}
    -\|\mbox{Im}(\kappa)\|_{\rz^3}^2}-\alpha^{-1}-\varepsilon_N
  \big\}\,. 
\end{equation*}
Hence $0$, eigenvalue of $\mathtt{H}(0)$, remains disjoint from the
essential spectrum of $\mathtt{H}(\kappa)$ for all $\kappa \in
\Sigma_{\nu _{\varepsilon_N}}$ (recall that \(\Sigma_{\nu
  _{\varepsilon_N}}\subset\Sigma_{\alpha^{-1}}\)) . 

Since $\mathtt{H}(\kappa)$ is an analytic family of type (B)
\cite[p.20]{RS4} in 
$\Sigma_{\nu_{\varepsilon}}$, $0$ is an eigenvalue of $\mathtt{H}(0)$
and moreover, $0$ remains disjoint from the essential spectrum of
$\mathtt{H}(\kappa)$, it follows that $0$ is an eigenvalue in the pure
point spectrum of $\mathtt{H}(\kappa)$ for all \(\kappa\in \Sigma_{\nu
  _{\varepsilon_N}}\) (reasoning as in
\cite[page 187]{RS4}). Let $\mathtt{P}(\kappa)$ be the projection onto
the eigenspace corresponding to the eigenvalue $0$ of the operator
$\mathtt{H}(\kappa)$. Then $\mathtt{P}(\kappa)$ is an analytic function
in $\Sigma_{\nu_{\varepsilon_N}}$ and for $\kappa \in
\Sigma_{\nu_{\varepsilon_N}}$ and $\kappa_0 \in \mathbb{R}$ we have 
\begin{equation*}
  \mathtt{P}(\kappa+ \kappa_0)= W(\kappa_0) \mathtt{P}(\kappa) W(-\kappa_0)\,.
\end{equation*}
Here we used that $W(-\kappa_0)$ is a unitary operator. The result of
the lemma
follows by applying Lemma~\ref{lem:O'Connor} below to
$\tilde{W}(\theta):=e^{i 
  \theta \kappa \cdot \bx}$ with $\kappa \in \mathbb{R}^3$, 
$\|\kappa\|_{\rz^3}= \nu_{\varepsilon_{N}}$, and $\theta\in\{z \in \mathbb{C}\,|\,
|\mbox{Im}(z)|<1\}$. 
Notice that $\tilde{W}(\theta)=W(\theta
\kappa)$ and that the projection
$\mathtt{\tilde{P}}(\theta):=\mathtt{P} (\theta \kappa)$ is analytic
and satisfies 
$\mathtt{\tilde{P}}(\theta+ \theta_0
)=\tilde{W}(\theta_0)\mathtt{\tilde{P}}(\theta)\tilde{W}(-\theta_0)$
for $\theta_0 \in \mathbb{R}$.

This finishes the proof of (i). 

To prove (ii), we can work directly with the Hartree-Fock equation,
since, from (i), the function \(K_{\gamma^{\rm HF}}\varphi\) is
exponentially decaying. Therefore, let
\begin{align}
  \label{eq:q-non-pert}
   \mathfrak{q}[u,v]=(u,h_{\gamma^{\rm HF}}v)-\varepsilon(u,v)
\ \text{ for }\ u,v \in
  H^{1/2}(\mathbb{R}^3)\,,
\end{align}
and note that, by assumption, \(0\) is an eigenvalue for the corresponding operator
(\(\varphi\) is an eigenfunction).
Define, for \(\kappa\in\rz^3\), 
\begin{align}
  \label{eq:holo-general}\nonumber
  \mathfrak{q}(\kappa)[u,v]&=\mathfrak{q}[W(-\kappa)u,W(-\kappa)v]\\
  &=\mathfrak{t}(\kappa)[u,v]-\mathfrak{v}[u,v]+\mathfrak{b}_{\gamma^{\rm
      HF}}(\kappa)[u,v]-\varepsilon(u,v)\,,
\end{align}
with \(W(\kappa)\) and \(\mathfrak{t}(\kappa)\) as before (but now on \(H^{1/2}(\rz^3)\)),
see \eqref{eq:sectQ}, 
and 
\begin{align}
  \label{eq:K-kappa}
  \mathfrak{b}_{\gamma^{\rm
      HF}}(\kappa)[u,v]=\alpha(u,R_{\gamma^{\rm HF}}v)
  -\alpha(u,K_{\gamma^{\rm HF}}(\kappa)v)\,,
\end{align}
where
\begin{align}
  \label{eq:K-kappa-kernel}
  K_{\gamma^{\rm HF}}(\kappa)(\bx,\by)
  =\sum_{j=1}^{N}\frac{\varphi_j(\bx)e^{i\kappa\bx}
  e^{-i\kappa\by}\overline{\varphi_j(\by)}}{|\bx-\by|}\,.
\end{align}
Using (i) of the proposition (exponential decay of the Hartree-Fock orbitals
\(\{\varphi_{j}\}_{j=1}^{N}\)) one now proves that
\eqref{eq:K-kappa-kernel} extends to a holomorphic family of
Hilberts-Schmidt operators in \(\Sigma_{\nu_{\varepsilon_N}}\). One
can now repeat the reasoning in the 
proof of (i) to obtain the stated exponential decay of \(\varphi\).

\end{pf*}
\begin{lemma}{\rm (\cite[p.\ 196]{RS4})}\label{lem:O'Connor}
Let $W(\kappa)=e^{i \kappa A}$ be a one-parameter unitary group (in
particular, $A$ is self-adjoint) and let $D$ be a connected region in
$\mathbb{C}$ with $0 \in D$. Suppose that a projection-valued analytic
function $P(\kappa)$ is given on $D$ with $P(0)$ of finite rank and so
that 
\begin{align*}
  W(\kappa_0) P(\kappa) W(\kappa_0)^{-1}=P(\kappa + \kappa_0) 
  \text{ for }\kappa_0 \in \mathbb{R} \text{ and }\kappa,
  \kappa+\kappa_0 \in D\,.  
\end{align*}
Let $\psi \in \mbox{\rm Ran}(P(0))$. Then the function $\psi(\kappa)=
W(\kappa) \psi$ has an analytic continuation from $D \cap \mathbb{R}$
to $D$. 
\end{lemma}

\appendix
\section{Some useful lemmata}\label{app:lemmata}

\begin{lemma}\label{regfj} 
Let $\Omega$ be an open subset of $\R^3\setminus \{0 \}$ with
smooth boundary
and let $f_1, f_2\in H^k(\Omega)$ for some $k
\geq 1$.  

Then the function 
\begin{equation*}
  F(\bx):=\int_{\R^{3}}\frac{f_1(\by)f_2(\by)}{|\bx-\by|}d\by 
\end{equation*}
belongs to $C^k(\Omega)$ if $k \geq 2$, while if $k=1$, it belongs to
$W^{1,p}(\Omega)$ for all $p \geq 1$, and hence to $C(\Omega)$.  
\end{lemma}

\begin{proof} We are going to prove the following equivalent
  statement. If $k\geq 2$, $\chi F \in C^{k}(\R^3)$ for all
  $\chi \in C^{\infty}_{0}(\Omega)$, while if $k=1$, $\chi F \in
  W^{1,p}(\R^3)$ for all $p \geq 1$ and $\chi \in
  C^{\infty}_{0}(\Omega)$. 

Fix $\chi \in C^{\infty}_{0}(\Omega)$ and take $\tilde{\chi} \in
C^{\infty}_{0}(\Omega)$ verifying $\tilde{\chi} \equiv 1$ on
$\supp\,\chi$ and such that there is a strictly positive distance
between $\supp\,\chi$ and $\supp\,(1-\tilde{\chi})$. We write $\chi
F(\bx) = \chi F_1(\bx)+\chi F_2(\bx)$ with  
\begin{equation*}
  F_1(\bx)
  =\int_{\R^{3}}\frac{\tilde{\chi}(\by)f_1(\by)f_2(\by)}{|\bx-\by|}\,d\by 
 \; \mbox{ and } \; 
 F_2(\bx)
  =\int_{\R^{3}}(1-\tilde{\chi}(\by))\frac{f_1(\by)f_2(\by)}{|\bx-\by|}\,d\by\,.  
\end{equation*} 
The term $\chi F_2$ is clearly in $C^{\infty}(\R^3)$. For the
other term we use Young's inequality: if $f \in L^{p}(\R^3)$ and $g
\in L^q(\R^3)$ then 
\begin{equation}\label{St1}
  \| f*g\|_{r} \leq C \| f\|_{p} \|g\|_{q} \; \; \mbox{ with
  }1+\frac{1}{r}=\frac{1}{p}+\frac{1}{q}\,. 
\end{equation}
Moreover, if $1/p+1/q=1$ then $f*g$ is continuous (see \cite[Lemma
2.1]{Tartar}). Let $\alpha \in \nz_{0}^3$ with $|\alpha|\leq k$. Then   
\begin{equation}\label{aa}
  |D^{\alpha} (\chi F_1) (\bx)| \leq
  \sum_{\substack{\beta_1+\beta_2=\alpha,\\ \beta_1,\beta_2 \in
    \nz_{0}^3}} |D^{\beta_1} \chi(\bx)| \Big|\int_{\R^3} \frac{1}{|\bx
  -\by|} D^{\beta_2} (\tilde{\chi} f_1 f_2)(\by) d\by
  \Big|\,. 
\end{equation}
If $f_1,f_2 \in H^{k}(\Omega)$, $k \geq 2$, then $D^{\beta_2}
(\tilde{\chi} f_1 f_2) \in L^{5/3}(\R^3)$ for all
$\beta_2$ as in \eqref{aa}. From \eqref{St1}, \eqref{aa} and
$\tilde{\chi}/|\cdot| \in L^{5/2}(\R^3)$ it follows that $D^{\alpha}
(\chi F_1)$ is continuous and, since $\alpha$ is arbitrary,
that $\chi F \in C^{k}(\R^3)$.  

If $f_1, f_2 \in H^{1}(\Omega)$ then $\partial
(\tilde{\chi} f_1 f_2) \in L^{3/2}(\R^3)$ and from
\eqref{St1} we get (only) that $\partial (\chi F) \in L^{p}(\R^3)$
for all $p\geq 1$. It then follows that $F \in W^{1,p}(\Omega)$ for
all $p\geq 1$ and therefore (by the Sobolev imbedding theorem) $F \in
C(\Omega)$.  
\end{proof}

\begin{lemma}\label{lem:HS}
Let, for \(Z\alpha<2/\pi\),  $h_{0}$ be the self-adjoint operator
defined in (9), and let $\Lambda_{-}(\alpha)$ be the projection onto
the pure point spectrum of $h_{0}$.  

Then the operator $\Lambda_{-}(\alpha) h_{0} \Lambda_{-}(\alpha)$ is
Hilbert-Schmidt. 
\end{lemma}
\begin{proof} Let $\epsilon>0$ be such that $Z\alpha (1+\epsilon)\leq
  2/\pi (1-\epsilon)$. We are going to prove that there exists a
  constant $M = M(\epsilon)$ such that 
\begin{equation}\label{uu}
  h_0\geq\frac{1}{M+2\alpha^{-1}}P( -\Delta-\frac{C}{|\cdot|})P\,, 
\end{equation}
with $C=Z\alpha (M+2\alpha^{-1})(1+1/\epsilon)$ and $P=
\chi_{[0,M]}(T(\bp))$. The claim will then follow from \eqref{uu}
since 
\begin{equation*}
  \Tr\big([h_0]_{-}\big)^2 \leq \frac{1}{(M+2\alpha^{-1})^2} \Tr\big([-\Delta-
  \frac{C}{|\cdot |} ]_{-}\big)^{2} < \infty\,. 
\end{equation*}
The last inequality follows since the eigenvalues of \(-\Delta-C/|\cdot|\)
are \({}-C^2/4n^2, n\in\nz\), with multiblicity \(n^2\).

We now prove \eqref{uu}. For \(\epsilon>0\) and any projection \(P\) (with
\(P^{\perp}={\mathbf 1}-P)\), we have that
\begin{align}\label{cauchy-s}\nonumber
  h_0&=P h_0 P + P^{\perp} h_0 P^{\perp} - P \frac{Z \alpha}{|\cdot
  |} P^{\perp}- P^{\perp} \frac{Z \alpha}{|\cdot |} P\\ 
  &\geq  P (h_0-\frac{1}{\epsilon} \frac{Z \alpha}{|\cdot |}) P +
  P^{\perp} (h_0-\epsilon \frac{Z\alpha}{|\cdot |}) P^{\perp}\,.
\end{align}
By a direct computation one sees that there
exists a constant $M=M(\epsilon)$ such that $T(\bp) \geq M$ implies
$T(\bp)\geq (1-\epsilon) |\bp|$ and $T(\bp) \leq M$ implies
$T(\bp)\geq \frac{1}{M+2\alpha^{-1}} (-\Delta)$. Hence, with this choice
of \(M\) and 
$P=\chi_{[0,M]}(T(\bp))$, \eqref{cauchy-s} implies that
\begin{equation*}
  h_{0} \geq P\big[\frac{1}{M+2\alpha^{-1}} (-\Delta)-(1+\epsilon^{-1})
  \frac{Z \alpha}{|\cdot |}\big]P + P^{\perp}\big[(1-\epsilon) \sqrt{-\Delta}
  -(1+\epsilon)\frac{Z\alpha}{|\cdot|}\big]P^{\perp}\,. 
\end{equation*}
The inequality \eqref{uu} follows directly by the choice of
$\epsilon$. 
\end{proof}

\section{Pseudodifferential operators}
\label{app:pseudo}

In this appendix we collect facts needed from the calculus of
pseudodifferential operators (\(\psi\)do's) (for references, see
e.g. \cite{Ho2Classics} or \cite{SR}).

Define the standard (H\"{o}rmander) symbol class \(S^{\mu}(\R^n)\),
\(\mu \in\R\),
to be the set of 
functions \(a\in C^{\infty}(\R^n_x\times\R_\xi^n)\) satisfying
\begin{align}
  \label{eq:symbEst}
  \big|\partial_x^\alpha\partial_\xi^\beta a(x,\xi)\big|
  \le C_{\alpha,\beta}(1+|\xi|^2)^{(\mu-|\beta|)/2}
  \ \text{ for all } (x,\xi)\in \R^n_x\times\R_\xi^n\,.
\end{align}
Here, \(\alpha, \beta\in\N^n\) and
\(|\alpha|=\alpha_1+\dots+\alpha_n\). 
Furthermore,
\(S^{\mu}(\R^n)\subset S^{{\mu}'}(\R^n)\) for \(\mu\le \mu'\). We denote
\(S^{\infty}(\R^n)=\cup_{\mu\in\R}S^{\mu}(\R^n)\)
and \(S^{-\infty}(\R^n)=\cap_{\mu \in\R}S^{\mu}(\R^n)\). Finally, note that
\(ab\in S^{\mu_1+\mu_2}(\R^n)\), \(\partial_x^\alpha\partial^\beta_\xi a
\in S^{\mu_1-|\beta|}(\R^n)\) when \(a\in S^{\mu_1}(\R^n), b\in S^{\mu_2}(\R^n)\). 

A symbol \(a\in S^{\mu}(\R^n)\) defines a linear operator \(A={\rm
  Op}(a)\in:\Psi^{\mu}\) (`pseudodifferential operator of order \(\mu\)') by
\begin{align}
  \label{eq:defOp}
  [{\rm Op}(a)u](x)=(2\pi)^{-n}\int_{\R^n}e^{ix\cdot\xi}a(x,\xi)\hat{u}(\xi)\,d\xi\,,
\end{align}
where \(\hat{u}\) is the Fourier-transform of \(u\). The operator
\(A\) is well-defined on the space \(\mathcal{S}(\R^n)\) of
Schwartz-functions; it extends 
by duality to \(\mathcal{S}'(\R^n)\), the space of tempered
distributions.  
Note that for
\begin{align}
  \label{eq:PDO}
  a(x,\xi)=\sum_{0\leq|\alpha|\le \mu}a_\alpha(x)\xi^\alpha
\end{align}
(with \(a_\alpha\) smooth and with all derivatives bounded, i.e.,
\(a_{\alpha}\in\mathcal {B}(\R^n)\)),
\(A={\rm Op}(a)\in\Psi^{\mu}\) is the partial differential operator given by
\begin{align}
  \label{eq:actionPDO}
  [{\rm Op}(a)u](x)=\sum_{0\leq |\alpha|\le \mu}a_\alpha(x)D^\alpha u(x)\,.
\end{align}
Note also that, with \(a=a(x)\) and \(b=b(\xi)\), 
\begin{align*}
 [{\rm
  Op}(a)u](x)=a(x)u(x)\  \text{ and }\ \widehat{[{\rm
    Op}(b)u]}(\xi)=b(\xi)\hat{u}(\xi)\,. 
\end{align*}
If \(a\in S^{\mu}(\R^n)\), then \({\rm Op}(a)\), defined this way, maps
\(H^k(\R^n)\) continuously into \(H^{k-\mu}(\R^n)\) for all
\(k\in\R\). Here, \(H^k(\R^n)\) 
is the Sobolev-space of order \(k\), consisting of
\(u\in\mathcal{S}'(\R^n)\) for which
\begin{align}
  \label{eq:defSob}
  \|u\|_{H^k(\R^n)}^2:=\int_{\R^n}|\hat{u}(\xi)|^2(1+|\xi|^2)^k\,d\xi
\end{align}
is finite; this defines the norm on \(H^k(\R^n)\). We denote 
\begin{align*}
   H^\infty(\R^n)=\bigcap_{k\in\R}H^k(\R^n)\,, \quad
   H^{-\infty}(\R^n)=\bigcup_{k\in\R}H^k(\R^n)\,.
\end{align*}
In particular, symbols in \(S^0(\R^n)\) define bounded
operators on \(L^2(\R^n)=H^0(\R^n)\). Furthermore, operators defined
by symbols in \(S^{-\infty}(\R^n)\) maps any \(H^k(\R^n)\) into
\(H^{\infty}(\R^n)\); such operators are called `smoothing'. 

We need to compose \(\psi\)do's. There exists a composition \(\#\) of
symbols,
\begin{align}
  \label{eq:hash}
  \# : S^{\mu_1}(\R^n)\times S^{\mu_2}(\R^n)&\to S^{\mu_1+\mu_2}(\R^n)\\
   (a,b)&\mapsto a\#b\,,
\end{align}
such that \({\rm Op}(a){\rm Op}(b)={\rm Op}(a\#b)\). It is given by
\begin{align}  \label{eq:actionSharp}
  (a\#b)(x,\xi)=\frac{1}{(2 \pi)^n}\int_{\R^n\times \R^n} e^{- i y
    \cdot \xi}a(x, \xi - \eta) b(x-y,\eta) \, dy d\eta\,.
\end{align}
Here, the integral is to be understood as an oscillating integral.

The symbol \(a\#b\) has the expansion
\begin{align}
  \label{eq:expansionSymb}
  a\#b\sim
  \sum_{\alpha}\frac{i^{-|\alpha|}}{\alpha!}(\partial^\alpha_xa)(\partial^\alpha_\xi b)\,.  
\end{align}
Here, `\(\sim\)' means that for all \(j\in\N\), 
\begin{align}
  \label{eq:asymp}
  a\#b - \sum_{|\alpha|<
    j}\frac{i^{-|\alpha|}}{\alpha!}(\partial^\alpha_xa)(\partial^\alpha_\xi b) 
  \in S^{\mu_1+\mu_2-j}(\R^n)
\end{align}
(recall that \((\partial^\alpha_xa)(\partial^\alpha_\xi b )\in
S^{\mu_1+\mu_2-|\alpha|}\)). 
One easily sees that the
composition is associative.
\begin{proposition}\label{corcom}
If $a \in S^{m_1}(\rz^n)$, $b \in S^{m_2}(\rz^n)$ then the symbol associated to
$[{\rm Op}(a), {\rm Op}(b)]$ belongs to $S^{m_1+m_2-1}(\rz^n)$. 
\end{proposition}

In particular, if \(\phi_1, \phi_2\in {\mathcal B}^{\infty}(\R^{n})\) (the smooth functions with
bounded derivatives) with 
\(\supp\,\phi_1\cap\supp\,\phi_2=\emptyset\) and \(a\in S^{\mu}(\R^n)\),
\(a(x,\xi)=a(\xi)\), then \(\phi_1\#a\#\phi_2\sim0\), and so, with $A := {\rm Op}(a)$,
$$
\phi_1 A \phi_2 = {\rm Op}(\phi_1){\rm Op}(a){\rm Op}(\phi_2)
$$
is smoothing. 

\begin{acknowledgement}
The authors wish to thank Heinz Siedentop for useful
discussions.  Support from the EU IHP network 
{\it Postdoctoral Training Program in Mathematical Analysis of
Large Quantum Systems},
contract no.\
HPRN-CT-2002-00277, and from the Danish Natural Science Research
Council, under the grant Mathematical Physics and Partial Differential
Equations, is gratefully acknowledged. 
T\O S wishes to thank the Department of Mathematics, LMU Munich, 
for its hospitality in the spring of 2007.
\end{acknowledgement}

  %      \bibliographystyle{amsplain}
 %           \bibliography{refRelHF}

\providecommand{\bysame}{\leavevmode\hbox to3em{\hrulefill}\thinspace}
\providecommand{\MR}{\relax\ifhmode\unskip\space\fi MR }
\providecommand{\MRhref}[2]{%
  \href{http://www.ams.org/mathscinet-getitem?mr=#1}{#2}
}
\providecommand{\href}[2]{#2}

\end{document}